    \documentclass[aps,prl,reprint,longbibliography, groupedaddress,superscriptaddress,amsmath,amssymbol,stfloats,notitlepage,footinbib]{revtex4-2}
    \usepackage[english]{babel} 
    \usepackage{amsfonts}
   \usepackage{hyperref}
    \usepackage{amsmath}
    \usepackage{fullpage}
    \usepackage{amsmath}
    \usepackage{mathtools}
    \usepackage{bbm}
    \usepackage[mathscr]{eucal}
    \usepackage{amssymb}%
    \usepackage{graphicx}
    \usepackage{xcolor}
    \usepackage[normalem]{ulem}
    \providecommand{\U}[1]{\protect\rule{.1in}{.1in}}
    \newcommand{\identity}{\mathbbm{1}} %\id
    \newcommand{\id}{\identity}
    \newcommand{\NN}{\mathcal{N}}
    \newcommand{\MM}{\mathcal{M}}

    \newcommand{\tr}{\mathrm{tr}}
    \newcommand{\ess}{\mathrm{ess\, sup}\,}
    \newcommand{\ket}[2]{\mathinner{|#2\rangle}_{\hspace{-0.1em} #1}}
    \newcommand{\ketbra}[3]{\mathinner{|#2\rangle\langle #3|}_{#1}}
    \newcommand{\ex}[1]{\mathrm{e}^{#1}}
    \newcommand*\circled[1]{\raisebox{.5pt}{\textcircled{\raisebox{-.9pt} {#1}}}
    }
    
    \newcommand{\EE}{\mathbb{E}}
    \newtheorem{theorem}{Theorem}

    \newtheorem{lemma}[theorem]{Lemma}

    \newenvironment{proof}[1][Proof]{\noindent\textbf{#1.} }{\ \rule{0.5em}{0.5em}}
    \usepackage{tikz}
    \usetikzlibrary{}
    \usepgfmodule{oo}

    \usepackage[percent]{overpic}

\begin{document}
    
    \title{\textbf{Ultimate limits for quickest quantum change-point detection}}
    
    \author{Marco Fanizza}
    \email{marco.fanizza@uab.cat}
    \affiliation{F\'{\i}sica Te\`{o}rica: Informaci\'{o} i Fen\`{o}mens Qu\`{a}ntics, Departament de F\'{i}sica, Universitat Aut\`{o}noma de Barcelona, ES-08193 Bellaterra (Barcelona), Spain}
    \author{Christoph Hirche}
    \email{christoph.hirche@gmail.com}
    \affiliation{Center for Quantum Technologies, National University of Singapore \\
                      Zentrum Mathematik, Technical University of Munich, 85748 Garching, Germany}
    \author{John Calsamiglia}
    \email{john.calsamiglia@uab.cat}
    \affiliation{F\'{\i}sica Te\`{o}rica: Informaci\'{o} i Fen\`{o}mens Qu\`{a}ntics, Departament de F\'{i}sica, Universitat Aut\`{o}noma de Barcelona, ES-08193 Bellaterra (Barcelona), Spain}

    \begin{abstract}
    {
    Detecting abrupt changes in data streams is crucial because they are often triggered by events that have important consequences if left unattended. Quickest change point detection has become a vital sequential analysis primitive that aims at designing procedures that minimize the expected detection delay of a change subject to a bounded expected false alarm time. We put forward the quantum counterpart of this fundamental primitive on streams of quantum data. We give a lower-bound on the mean minimum delay when the expected time of a false alarm is asymptotically large, under the most general quantum detection strategy, which is given by a sequence of adaptive collective (potentially weak) measurements on the growing string of quantum data. In addition, we give particular strategies based on repeated measurements on independent blocks of samples, that asymptotically attain the lower-bound, and thereby establish the ultimate quantum limit for quickest change point detection. Finally, we discuss online change point detection in quantum channels.
}
    \end{abstract}
    
    \maketitle

 {   
 The detection of sudden changes of a stochastic random
 variable is one of the most fundamental problems
 in statistical analysis with a wide range of applications. For instance, in the context of industrial process monitoring ~\cite{hawkins2003changepoint,lai1995sequential}, detecting sudden changes in sensor readings can help identify faults or equipment malfunctions that could lead to downtime or accidents if not addressed in a timely manner.  In medical  sciences \cite{chen2011parametric,rosenfield2010change}, change point detection can be used to detect the onset of an infectious disease outbreak or the onset of a medical emergency while monitoring the vital signs of a patient. Similar examples can be found in climate research~\cite{gallagher2013changepoint,reeves2007review}, cybersecurity \cite{kurt_distributed_2018} and robotics~\cite{niekum2015online, ranganathan2012pliss} to name a few. Change-point detection has become a field of its own  in classical   statistical analysis~\cite{basseville1993,chen2011parametric,poor2008,tartakovsky2014} with large activity on the fundamental side, establishing optimal estimators and trade-off regions, in
diverse idealized settings, and on the applied side developing statistical and machine learning techniques that operate with real life data streams. }

    Recently, the concept of change point detection has been generalized to the quantum world~\cite{Akimoto2011,QCP1,QCP2,QCPonline}. Here, we have a device outputting quantum states. By default
    this device will output a certain state $\rho$, but from a given
    (random) point of time it will start producing the state
    $\sigma$. The goal is to identify this change point. In~\cite{Akimoto2011,QCP1,QCP2} the problem has been considered as an instance of hypothesis testing, where one collects a fixed number of quantum states and then tries to determine if and where a change point occurred. Since one requires the full sequence, this is usually considered offline change point detection.   For the special case of pure quantum states, solutions where given for the identification of the position of the change, optimizing the mean probability of error~\cite{QCP1} and the probability of  unambiguous identification~\cite{QCP2}.
    For any practical application, however, one usually
    cannot wait for the entire sequence to be collected (it  could potentially be infinite). Therefore,
    it is well motivated to consider online detection or quickest change point detection: an algorithm that samples every copy sequentially and fires an alarm 
as soon as it detects the change.   In this scenario, the natural quantities to
    consider are the time delay in detecting a change point
    versus the risk of a false alarm, i.e. falsely detecting a change when none
    has happened.
    In classical statistical analysis the most studied such   algorithm is Page's cumulative sum (CUSUM) algorithm~\cite{page54}. Apart from    its computational simplicity, one of its most important  features is its optimality under certain risk criteria, as shown first by Lorden~\cite{lorden71} in the asymptotic setting and in  \cite{moustakides1986optimal,ritov1990decision} in the finite regime. 
    In~\cite{QCPonline} online strategies for quantum change point detection
    have been considered in the restricted scenario of  pure states where unambiguous (local) identification is possible. In this case, the post-change state $\sigma$ does not lie in the support of the default state $\rho$, and one can find a suitable measurement for which one of the outcomes can only be triggered by the post-change state, and thereby guarantees the absence of false alarms while keeping a finite mean detection delay. However, for realistic (mixed) states there is a trade-off between the false alarm time and the detection delay. In order to optimize this trade-off  one needs to consider general sequential quantum strategies. This class of strategies has been recently studied in the context of sequential hypothesis of  quantum states~\cite{vargas2021quantum,li2022optimal} and channels~\cite{SQCD}.
    
As the main contribution  of this work we provide the ultimate quantum limit for quickest change point detection: we give a lower bound  to the mean detection delay that can be reached by the class of most general detection strategies with a given bounded (asymptotically large) expected false alarm time; and we provide a  quantum version of the CUSUM algorithm, called QUSUM, and show that it asymptotically attains the aforementioned lower bound. In particular, this algorithm uses only $l$-local (projective) measurements on the incoming quantum states, still asymptotically it cannot be outperformed by even a sequence of possibly weak, collective and adaptive quantum measurements. 
    
  %  The crucial ingredient for the optimality proof is the use of a strong converse property in quantum hypothesis testing~\cite{Ogawa2000}, to verify that a certain condition that implies a converse bound in a (possibly non i.i.d.) change point problem~\cite{lai1995sequential,tartakovsky2014} has to be satisfied by any candidate quantum algorithm.
  
       The paper is structured as follows: we first present the problem in its simplest form and state the two main results, which are then proven in dedicated sections. 
    Afterwards, we comment on the implications of our results for the problem of change point detection in sequences of quantum channels. We conclude by mentioning open problems.
    
    \textit{Setting and results}:--- {The change point sequence is a sequence of $d$-dimensional states $\{\rho^{(n)}\}$, $n=1,2,...$, such that if $n\leq \nu$, $\rho^{(n)}=\rho$ and if $n>\nu$, $\rho^{(n)}=\sigma$. {At each step $n$, the algorithm receives a copy of the state $\rho^{(n)}$. The latter is then measured together with the current state of previously received copies by a joint%, %possibly weak, 
    quantum measurement, whose outcome determines whether to continue or to stop at step $n$ and emit an alarm that signals that the change has occurred.}} Let's call $T$ the random variable corresponding to the alarm time $n$ at which stopping occurs, for some given strategy.  { Let $\EE_{\infty/\nu}$ denote expectation values with respect to some measurement acting sequentially on a sequence of copies of $\rho$ (the change point never happens, $\nu=\infty$) or the change point sequence for a specific finite $\nu$. Similarly, the probability of an event $E$ is denoted by $P_{\infty/\nu}(E)$. We leave the precise algorithm implicit, but it should be clear that it could be any sequential quantum measurement on the sequence $\{\rho^{(n)}\}$, with outcomes described by discrete random variables $\{X_1,...,X_n\}${, where $X_i$ is the outcome of the measurement realized after getting $\rho^{(i)}$. We also use $X^n$ to denote the vector of random variables $\{X_1,...,X_n\}$ and $x^n$ to refer to the vector of values $\{x_1,...,x_n\}$ they assume.  We define the mean false alarm time as 
    \begin{align}
    \bar{T}_{\text{FA}} = \EE_{\infty}[T]. 
    \end{align}
    {We will consider families of strategies which have $\bar{T}_{\text{FA}}$ larger than a constant. Having a large expected false alarm time is desirable to avoid stopping early, i.e. before the change ($T<\nu$). In addition, in order to quantify the response time or delay ($T-\nu> 0$), we define the so-called worst-{worst case} mean delay as~\cite{tartakovsky2014} 
    \begin{align}
    \bar\tau^{\star} &:=\sup_{\nu\geq 0} \sup_{\tiny\substack{x^\nu \mathrm{ with }\\P_{\infty}(X^\nu=x^\nu)>0}}\EE_{\nu}[T-\nu | T > \nu, X^\nu=x^\nu].
    \end{align}
     {This figure of merit considers the worst mean delay over all possible locations of the change point and over all possible measurement outcomes before the change {\footnote{If POVMs with non-discrete outcomes are allowed, then one may consider also $\bar\tau^{\star} :=\sup_{\nu\geq 0} \ess \EE_{\nu}[T-\nu | T > \nu, X_1,...,X_{\nu}]$, where the essential supremum means that the optimization is done a supremum over set of measurement outcomes of non-zero measure, where the measure of a set is given by its probability. For simplicity, in this paper we consider only measurement with discrete outcomes} %but the proofs can be extended with straightforward modifications to the case of measurement with non-discrete outcomes.}.} A less strict figure of merit is the worst mean delay over all possible locations of the change point:
    \begin{equation}
    \bar\tau :=\sup_{\nu\geq 0} \EE_{\nu}[T-\nu | T > \nu]\leq \bar\tau^{\star}. 
    \end{equation}
    }
    
   In the following we always assume $\mathrm{supp}\, \sigma \subseteq \mathrm{supp}\, \rho$.   Otherwise, as in the pure case discussed above, there exist a projector $\Pi$ such that $\tr[\Pi\rho]=0$ and $\tr[\Pi\sigma]=c>0$, therefore the change can be detected with high probability in finite time, and no false alarms (infinite $\bar{T}_{\text{FA}})$. This also implies that the two entropic quantities, which will play a prominent role here, $D(\sigma||\rho)=\tr[\sigma(\log \sigma-\log\rho)]$ and $D_{\max}(\sigma||\rho)=\inf\{\lambda\geq 0 : \sigma\leq 2^\lambda \rho\}$ are bounded.
      For strategies with fixed false alarm time $\bar{T}_{\text{FA}}$, optimal strategies are those which minimize $\bar{\tau}^*$. We show that the asymptotic behaviour of the optimal $\bar{\tau}^*$ for large $\bar{T}_{\text{FA}}$ is}: \begin{equation}\bar{\tau}^*\sim \frac{\log\bar{T}_{\text{FA}}}{D(\sigma \| \rho )}.
    \end{equation}

%\jccc{ We show..***
%This implies that for large $\bar{T}_{\text{FA}}$ optimal strategies satisfies $\bar{\tau}^*\sim \frac{\log\bar{T}_{\text{FA}}}{D(\sigma \| \rho )}$}.

We prove this via two theorems, which respectively provide an upper bound and a lower bound to  $\bar\tau^\star$.    
    {
    \begin{theorem}[Achievability]
    Given a change point problem with two finite-dimensional states $\rho$ and $\sigma$, $D(\sigma||\rho)<\infty$, for any $\epsilon>0$, and $\bar{T}_{\text{FA}}$ large enough there is a QUSUM algorithm such that $\bar\tau^\star \leq \frac{\log\bar{T}_{\text{FA}}}{D(\sigma \| \rho )(1-\epsilon)}+O(1)$.
    \end{theorem}
    
     QUSUM is, as the name suggests, a quantum version of the classical CUSUM  algorithm, %which in turn is based on repeated sequential probability ratio tests (SPRT),
     see also~\cite{page54}. We show that the performance of QUSUM is asymptotically optimal:
    
    \begin{theorem}[Optimality]\label{optimtheo}
    Any algorithm for a change point problem with two finite-dimensional states $\rho$ and $\sigma$, $D(\sigma||\rho)<\infty$, with expected false alarm $\bar{T}_{\text{FA}}$ must satisfy $\bar\tau^\star \geq\bar\tau\geq (1-\epsilon)\frac{\log\bar{T}_{\text{FA}}}{D(\sigma \| \rho )}(1+o(1))$ for any $\epsilon>0$.
    \end{theorem}
    
In the following we give the proof of these theorems, and discuss some  generalizations thereof.
    
    \textit{Achievability}:---
    }
    We will prove the achievability in two steps. First we study the detection delay of simple algorithm that repeats the same measurement on each individual incoming state, an then extended it to the case  where repeated measurements are performed on blocks of a fixed number of copies. If a fixed POVM $\{ M_{x_{i}} \}$ is applied to the $i$-th state, outcome $x_{i}$ will appear with probability
\begin{align}
  p(x_i)=\tr[M_{x_i} \rho] \mbox{ or }   q(x_i)=\tr[M_{x_i} \sigma],
 \end{align}
 depending on the underlying state. 
 We can define the log-likelihood ratio and their partial sums, 
    \begin{align}
    Z_i = \log\frac{q(x_i)}{p(x_i)},\label{defzi} \; Z^n_j=\sum_{i=j}^n Z_i.
    \end{align}
     It can easily be seen that the mean of the first is given by the relative entropy
    \begin{align}
    \EE_q[Z_i] =\sum_{i}q_{i}\log\frac{p_{i}}{q_{i}}=: D(q \| p ).
    \end{align}
    where $\EE_{q}$ denote expectation values with respect to the probability distribution $q$.   Note that second quantity in \eqref{defzi} with $j=\nu+1$ gives the log-likelihood ratio of a sequence of i.i.d. outcomes $x^{n}$ sampled from either a source with change point at $\nu$ or from a source with no change: $\lambda^{(\nu)}_{n}:=\log\frac{P_{\nu}(X^{k}=x^{k})}{P_{\infty}(X^{n}=x^{n})}=Z^n_{\nu+1}$. 
 Following the CUSUM algorithm we can fix a threshold value $h$ and consider for each possible change point $j$ a stopping time 
    \begin{align}
    T_j = \min\{ n \geq j : Z_j^k \geq h \},
    \end{align}
    where we define $T_j=\infty$ if \mbox{$\{n\geq j:Z_j^n\geq h\}$} $=\emptyset$,
and given these, we define the CUSUM stopping time 
    \begin{align}\label{eq:stopbig}
    T^\star = \min_{j\geq 1} T_j.
    \end{align}
    From the above definition of $\lambda^{(\nu)}_{n}$ it follows that $T^*$ can be understood as the first time when, given the current measurement record, the probability of having had a change in the past is \mbox{$\ex{h}$-times} more likely than having no change ---see the %Figure \ref{fig:1} and 
    Supplementary Material (SM) for a commented picture of a classical CUSUM test. We can now use a result by Lorden (Theorem 2 in~\cite{lorden71}). %which states that the properties of an online change point algorithm can be deduced from the properties of a set of parallel open-ended SPRTs.
    In our notation, we have that if $P_{\infty}(T_1 < \infty) \leq \alpha$,  
    \begin{align}
    \bar{T}_{\text{FA}} = \EE_\infty[T^\star] \geq \frac1\alpha, \,\,\text{and}\,\,
    \bar\tau^\star \leq\EE_{0}[T_1].\label{upperbound0} 
    \end{align}
  %  Now following , by using Wald's identities{ ~\cite{wald2004sequential}, (see also Chapter 3 of~\cite{tartakovsky2014}),
  
The above premise holds since
    \begin{align}\label{boundfalse}
    &P_\infty(T_1<\infty)=\mathbb E_\infty \left[I_{ T_1<\infty}\right]=
    \mathbb E_0\left[\frac{p(x^{T_1})}{q(x^{T_1})} I_{ T_1<\infty}\right]=\nonumber\\ 
   & \EE_0\left[\ex{-Z^{T_1}_1}I_{ T_1<\infty}\right]
   %=\mathbb{E}_0\left[\ex{-h-x}I_{ T_1<\infty}\right]
   \leq \ex{-h} =: \alpha
    \end{align}
    where in the second equality we have used the change of measure in order switch the distributions on which the expectation value is computed: $ \EE_p[f(x)]=\sum_x p(x) f(x)=\sum_x q(x)\frac{p(x)}{q(x)} f(x)=\EE_q[\frac{p(x)}{q(x)} f(x)]$; and the last equality holds because at the stopping time $Z^{T_1}_1$ is necessarily larger than $h$. Finally, using Wald's identity~\cite{wald2004sequential} (see also Chapter 3 of~\cite{tartakovsky2014} and SM), $\EE_0[Z_1^{T_1}]=\EE_0[\sum_{i=1}^{T_1} Z_i]=\EE_0[Z_1] \EE_0[T_1]$:
    
\begin{align}\label{waldlim}
    &\EE_0[T_1] = \frac{\EE_0 [Z_1^{T_1}]}{\EE_0 [Z_1]}=\frac{h+\mathbb E[s]}{D(q\| p)} \rightarrow \frac{h}{D(q\| p)},
\end{align}
    when $h\rightarrow \infty$, where $s:=Z_1^{T_1}-h$ is the \textit{overshoot} and the limit holds since $Z_i<\infty$.
    }
    }
    Putting all the results together we get, using \eqref{waldlim} in \eqref{upperbound0}
    \begin{align}
    \bar\tau^\star &\leq \EE_0[T_1]= \frac{h}{D(q\| p)}+O(1)\leq \frac{\log\bar{T}_{\text{FA}}}{D(q\| p)}+O(1).
    \end{align}
    Optimizing over all measurements gives us the achievabile trade-off for this particular strategy, 
    \begin{align}\label{Eq:meas-RE-achievable}
    \bar\tau^\star \leq \frac{\log\bar{T}_{\text{FA}}}{D_M(\sigma \| \rho )}+O(1) \mbox{ as } \bar{T}_{\text{FA}}\rightarrow\infty
    \end{align}
  in terms of the measured relative entropy $D_M(\sigma \| \rho ):=\sup_{\{M_i\} \mathrm{POVM}} D( q\|p)$~\cite{petz88, Hayashi_2001}. Note that already projective measurements (PVM) achieve the measured relative entropy~\cite{berta2017variational}.
    
    %\begin{figure}[ht]
   % \includegraphics[scale=.15]{CUSUM_curves2.pdf}
    %\vspace{-1em}
   % \caption{CUSUM for Bernoulli trial (coin toss) with a bias $p=1/5$ that at time $k=\nu=10^4$ changes to bias $q=1/4$ (see also SM).  %CUSUM  keeps track of $Z_1^k$, updating its value at every time step by $Z_1^n=Z_1^{n-1}+Z(x_n)$ depending on the (random) outcome $x_n$.
   % Dark green and light stochastic curves show the random walk exhibited by $Z_1^n$ for different measurement sequences $x_1^n$.  The orange line shows the average trend, which is given by  $\EE_\nu(Z_1^n)=-n D(p\|q)$  for $n\leq \nu$, while at the change point the slope changes abruptly to $D(q\|p)$.  The algorithm stops and signals a change-point as soon as the log-likelihood exhibits a net increase  larger or equal than $h$ with respect to some point $j$ in the past: $Z_j^n=Z_1^n-Z_1^j\geq h$. Two scenarios are show-cased: $\circled{1}$ a large threshold value ($h=22$) reduces the chances of false alarms at the expense of long detection times; $\circled{2}$  a low threshold value ($h'=6$) can detect the change point with a small delay, but has a high risk of producing false alarms.}
   % \label{fig:1}
   % \end{figure}

    \begin{figure}[ht]
    \includegraphics[width=1\columnwidth]{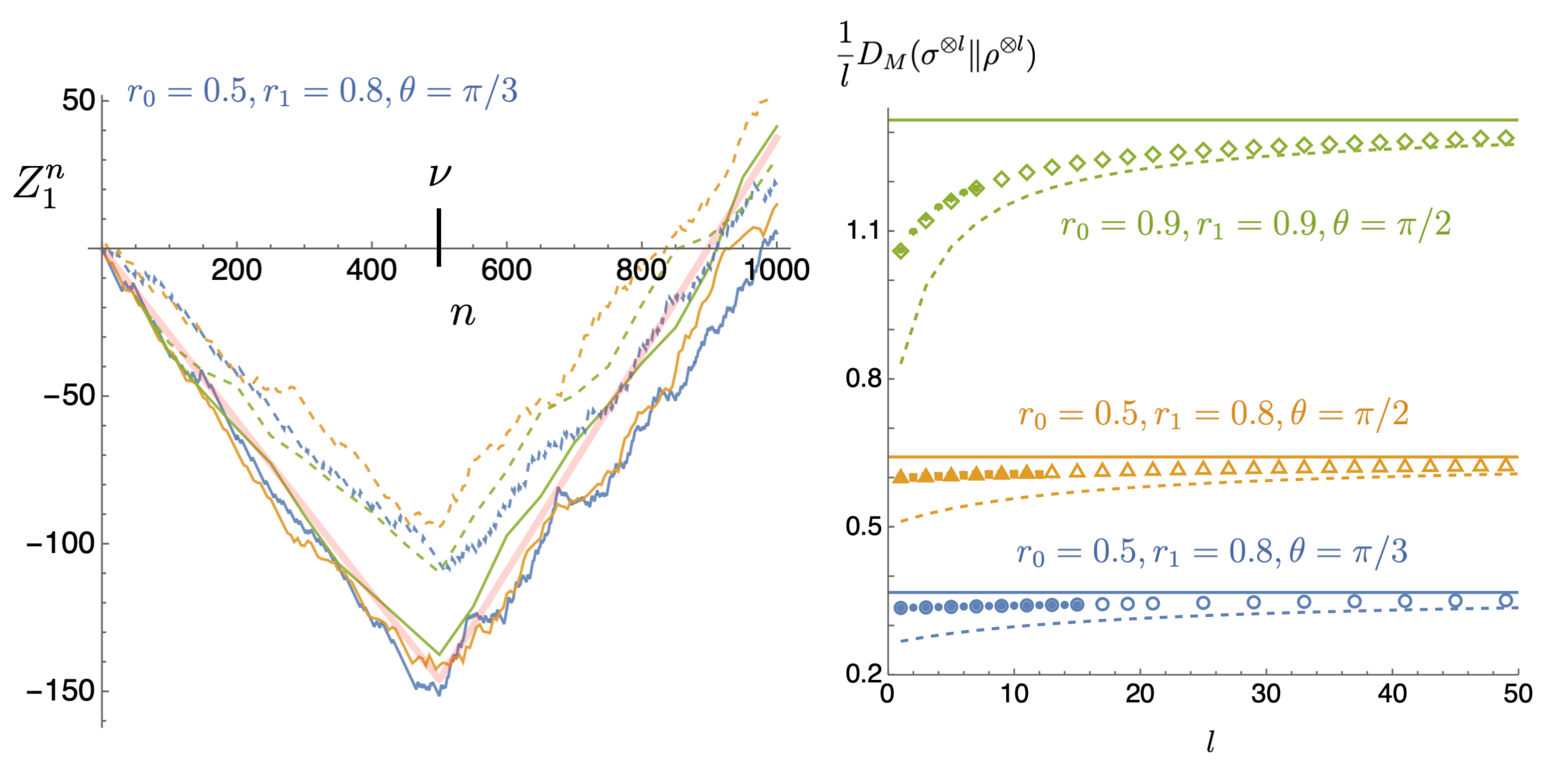}
    \caption{QUSUM for qubit states with Bloch vector of lenths $r_0$ ($\rho$) and  $r_1$ ($\sigma$) and relative angle $\theta$. Left: Log-likelihood stochastic trajectory for Hayashi's (dashed)  and $j$-angle-optimized (solid)     block-sampling strategies for block lengths $l=1,5,50$ (blue, orange, green). Larger $l$ approach the optimal rate given by the quantum relative entropy (thick red line). Right:  Measured relative entropy (per copy) for several strategies: Hayashi (dashed), $j$-angle-optimized (open markers), optimal from SDP (dots) and asymptotically attainable upper bound given by quantum relative entropy.
    }
    \label{fig:2}
    \end{figure}
    
    Now, more generally, instead of measuring each copy of $\rho^{(i)}$ separately, the QUSUM algorithm is based performing a joint measurement on blocks of  $l$ states which are either $\rho^{\otimes l}$ or $\sigma^{\otimes l}$ (assuming the change point happens at a multiple of $l$, see the SM for the general case). The above trade-off is now easily modified to $\bar\tau^\star \leq \frac{\log(\bar{T}_{\text{FA}}/l)}{\frac1l D_M(\sigma^{\otimes l} \| \rho^{\otimes l} )}+O(1).$
    
    If $\rho$ and $\sigma$ are states of a $d$-dimensional Hilbert space, Hayashi showed (Theorem 2 of~\cite{Hayashi_2001}) that for any $\sigma$ and $l$ there is a PVM $\{M^{(l)}_{x_i}\}$, depending only on $\rho$, such that if $p^{(l)}(i)=\tr[M^{(l)}_{x_i}\rho^{\otimes l}]$, $q^{(l)}(i)=\tr[M^{(l)}_{x_i}\sigma^{\otimes l}]$, $\forall \sigma$
    \begin{align}\label{convergencemeasured}
    D(\sigma\| \rho)-\tfrac{(d-1)\log (l+1)}{l}&\leq\frac1l D(q^{(l)} \| p^{(l)} )\leq D(\sigma\| \rho).
    \end{align}
   
    Choosing $l$ such that $\tfrac{1}{l}D(q^{(l)}||p^{(l)})\geq D(\sigma||\rho)(1-\epsilon)$ we obtain the statement of the theorem \footnote{Notice also that an upper bound on the sufficient size of the blocks to get the achievability in the theorem is $k=\tilde O(d/(\epsilon D(\sigma||\rho))$ (hiding logarithmic factors), with a better dependence on the dimension than optimal tomography, which requires $\Theta(d^2)$ copies~\cite{haah2017sample, o2016efficient}. In the SM we argue that for any fixed pair of states one can in fact choose $l$ depending only on relative entropies of the pair, therefore not explicitly on the dimension, based on~\cite{audenaert2012quantum}.}. Since we are considering $h\rightarrow \infty$ we can choose $l$ arbitrarily large and consider the limit $
    \lim_{l\rightarrow\infty} \frac1l D_M(\sigma^{\otimes l} \| \rho^{\otimes l} ) = D(\sigma \| \rho )$.
    This implies that in the asymptotic limit of large $h$ and $l$ we have
    \begin{equation}\label{tradeoff}
    \bar\tau^\star \leq \frac{\log\bar{T}_{\text{FA}}(1+o(1))}{D(\sigma \| \rho )}.
    \end{equation}
    {
    {In Fig.~\ref{fig:2} we illustrate QUSUM tests for qubit states, where we use the measurement of~\cite{Hayashi_2001}, and the enhanced class of $j$-angle-optimized measurements 
    that has a quicker convergence to the measured relative entropy, as the block size increases, as explained in SM.
    %that attain the measured relative entropy for finite block-sizes, as explained in the SM.
    The figure also shows the   the measured relative entropy for different block lengths, which determines the performance of block sampling strategies. The (optimal) measured relative entropy, whose approximation is computed by a semi-definite program (SDP) \cite{berta2017variational,  cvxquad}, is shown to agree with the values obtained for the  $j$-angle-optimized  measurement. {As shown in SM the latter measurement, as well as Hayashi's, are based on Schur sampling \cite{bacon_efficient_2006,harrow_phd,krovi_efficient_2019} and are hence efficiently implementable.}
    
 In addition, since the measurement in \cite{Hayashi_2001} achieving this bound does not depend on $\sigma$, we can also generalize this achievability result in the case where we have that the state after the change point is unknown and belonging to a finite family of states $\mathcal{S}$. In this case, we get asymptotically
    \begin{align}
    \bar\tau^\star \leq \frac{\log\bar{T}_{\text{FA}}(1+o(1))}{\min_{\sigma\in\mathcal{S}}D(\sigma \| \rho )}.
    \end{align}
    Guarantees in the case of infinite families can be obtained if there exists a suitable discretization of $\mathcal {S}$ (see SM for details).
   % {\color{purple}and if the same test is applied to a state $\sigma'$ not in $\mathcal S$ but such that $\min_{\sigma\in S}D(\sigma'||\rho)-D(\sigma'||\sigma)>0,$ we have asymptotically
   % \begin{align}
   % \bar\tau^\star \leq \frac{\log\bar{T}_{\text{FA}}}{\max_{\sigma\in\mathcal{S}}(D(\sigma'||\rho)-D(\sigma'||\sigma))}.
   % \end{align}
   % This allows to give guarantee on the false alarm in the case of unknown post-change state in a continous family, by choosing $\mathcal S$ as a discretization of the family.    
    }
    }
    
    \textit{Optimality}:---
    In this section we prove Theorem 2, which shows that no strategy can attain a better trade-off between the detection delay and the false alarm time than that given  by the quantum relative entropy.    To that end we first have to define the considered class of strategies. Since they can make use of general adaptive measurements,% (possibly weak), 
    it is necessary to specify how the state changes after the measurement, using quantum instruments. A quantum instrument is described by a set of completely-positive trace-non-increasing maps $\{\MM_x(\cdot)\}$, with $\sum_{x}\MM_x(\cdot)$ trace preserving.    For a fixed measurement outcome $x$ the {(normalized)}{}  post-measurement state is
    $\rho_x = \frac{\MM_x(\rho)}{\tr(\MM_x(\rho))}$, 
    where $\tr(\MM_x(\rho))$ corresponds to the probability of obtaining $x$ given the state $\rho$. In our setting, at each step $i$, we get a fresh copy of $\rho^{(i)}$, which is either $\rho$ or $\sigma$. %\jcc{We denote a restriction of a possible sequence of the outcomes  to the first $i$ elements as $x^i$}{}. 
    Let $\rho_{x^{i-1}}$ be the post-measurement state of the $(i-1)$-th step, then we apply the $i$-th quantum instrument $\MM^i$ (possibly depending on previous records $x^{i-1}$) as $\MM^i(\rho^{(i)}\otimes\rho_{x^{i-1}})$, 
    receiving a classical output $x_i$ and a new post-measurement state 
    $\rho_{x^{i}} = \frac{\MM^i_{x_i}(\rho^{(i)}\otimes\rho_{x^{i-1}})}{\tr\MM^i_{x_i}(\rho^{(i)}\otimes\rho_{x^{i-1}})}$. We denote the post-measurement states as $\rho^{(\infty)}_{x^{i}}$ if they originate from a sequence with no change point, and $\rho^{(\nu)}_{x^{i}}$ if they come from a sequence with change point $\nu$.
    We denote  $p(x_i | x^{i-1}) = \tr\MM^i_{x_i}(\rho\otimes\rho^{\infty}_{x^{i-1}})$, and $q^{(\nu)}(x_i | x^{i-1}) = \tr\MM^i_{x_i}(\sigma\otimes\rho^{(\nu)}_{x^{i-1}})$.
      We now define the local and cumulative log-likelihood ratios at step $i$ for a candidate change point $\nu$:
    \begin{align}
    Z_i^{(\nu)} = \log\frac{q^{(\nu)}(x_i | x^{i-1})}{p(x_i | x^{i-1})},\,\,\,\lambda^{(\nu)}_{n} =\sum_{i=\nu+1}^{n} Z_i^{(\nu)}.
    \end{align}
    
    Note that for a fixed sequence $x^i$ we can always write $p(x^i)=\tr M_{x^i} \rho^{\otimes i}$ for a joint measurement $\{ M_{x^i}\}$ giving a sequence of outcomes $x^i$. 
    
    While we still get a sequence of classical measurement outcomes as a result, these can now be highly correlated and the usual techniques for i.i.d. distributions do not longer apply. In the following we will make heavy use of a result initially stated by Lai~\cite{lai98} and reformulated in~\cite{tartakovsky2014}, which we adapt in a form which is applicable to our case, and give the proof in the SM for completeness~\footnote{Ref.~\cite{lai98} shows that the proof applies also when the probability in presence of a change point, conditioned on past events, can depend on the change point location. The statement of the theorem in~\cite{lai98} has an extra condition about the convergence of the log-likelihood which is not needed for our purposes and expresses the tradeoff only as a limit for large $\log \bar{T}_{\text{FA}}$ and $\epsilon\rightarrow 0$, while we state it with finite $\epsilon$}.
     
 \begin{theorem}
    For a change point model with log-likelihoods $Z_i^{(\nu)}$ and $\epsilon>0$, no strategies can exceed the trade-off given by $\bar\tau^*\geq (1-\epsilon) \frac{\log{\bar{T}_{\text{FA}}}}{I}
    (1+o(1))$, for large $\bar{T}_{\text{FA}}$, for any $I$ that satisfies the condition 
    
      \begin{align}\label{eq:condP}
    \lim_{n\rightarrow\infty} &\sup_{\nu\geq 0} \sup_{\tiny\substack{x^\nu \mathrm{ with }\\P_{\infty}(X^\nu=x^\nu)>0}} P^*_\nu(x^\nu) =0 
    \mbox{ where }\\
    P^*_\nu(x^{\nu})&:=P_{\nu}\left\{ \max_{ i\leq n }  \lambda^{(\nu)}_{\nu+i} \geq I(1+\epsilon) n \middle\vert X^\nu=x^\nu \right\}. \nonumber
    \end{align} 
    
    \end{theorem}
   In loose terms, the rate $I$ in this theorem has to be such that for any value $I'>I$ arbitrarily close to $I$, the stochastic trajectories exhibited by $\lambda^{(\nu)}_{\nu+i}$ (equivalent to those in the lhs of Fig. \ref{fig:2}) that exceed the value of $I' n$  between the change point $\nu$ and $\nu+n$ occur with a vanishing probability as $n$ increases. 
   
   The challenging art is to determine the smallest $I$ such that Eq.~(\ref{eq:condP}) holds for any $\epsilon>0$. For iid distributions the relative entropy $I=D(q||p)$ can be seen to satisfy this criterion.   In our case we need to find the rate $I$ considering that the underlying probability distribution can be produced by the most general kind of measurement strategy.  
 
    We start by getting rid of the supremum over $\nu$. 
    Note that all states up to position $\nu$ will be $\rho$ as in the case that there is no change. At position $\nu+1$ we will therefore try to discriminate between two states $\sigma \otimes \rho_{x^\nu} \quad\text{and}\quad \rho \otimes \rho_{x^\nu}$. It is now easy to see that for any $\nu$ and any measurement on the sequence of states, there also exists a measurement in the case of $\nu=0$ that results in the same probability distribution after the change point, consisting on simply preparing the state $ \rho_{x^\nu}$ and then applying the original strategy.  
    It follows that we can without loss of generality set $\nu=0$ and therefore also omit the essential supremum.

    {We will now bound $P^*_{0}$ based on the strong converse for quantum Stein's Lemma~\cite{Ogawa2000}, which states that if a sequence of binary tests $\{M_{0}^{(n)},M_1^{(n)}=\id-M_{0}^{(n)}\}$ is such that $\tr[M_1^{(n)}\rho^{\otimes n}]\leq \ex{-n( D(\sigma||\rho)+\delta))}$ for some $\delta>0$, then $\lim_{n\rightarrow\infty}\tr[M_1^{(n)}\sigma^{\otimes n}]=0$.} Let us denote the log-likelihood ratio of the outcome sequence $x^{i}$ for $\nu=0$ as $\lambda_{x^i}$. Define the set $ S_i := \{ x^i | \lambda_{x^i} \geq n I(1+\epsilon), \lambda_{x^j} < n I(1+\epsilon)\, \forall j<i \}. $
   
    We have the following chain of equalities, 
    \begin{align}
        &P^{(0)}\left(\max_{1\leq i\leq n} \lambda_{x^i}\geq n I(1+\epsilon)\right) \nonumber\\
        &= \sum_{\substack{x^n: \\ \max_{1\leq i\leq n} \lambda_{x^i}\geq n I(1+\epsilon)}} q^{(0)}(x^n) \nonumber\\
        &= \sum_{1\leq i \leq n} \sum_{x^i\in S_i} q^{(0)}(x^i) = \sum_{1\leq i \leq n} \sum_{x^i\in S_i} \tr M^{(i)}_{x^{i}}\sigma^{\otimes i} \nonumber\\
        &= \sum_{1\leq i \leq n} \sum_{x^i\in S_i} \tr [M^{(i)}_{x^{i}} \otimes \id^{\otimes (n-i)} \sigma^{\otimes n} ].
    \end{align}
    Defining the binary POVM $\{ \tilde M^i_0 , \tilde M^n_1 = \id - \tilde M^n_0 \}$, with $\tilde M^n_1 = \sum_{1\leq i \leq n} \sum_{x^i\in S_i} M^{(i)}_{x^{i}} \otimes \id^{\otimes (n-i)}$, we get $P^{(0)}\left(\max_{1\leq i\leq n} \lambda_{x^{i}}\geq n I(1+\epsilon)\right)= \tr \tilde M^n_1 \sigma^{\otimes n}$. Also, since   
    $q(x^i)=\ex{\lambda_{x^j}} p(x^i)\geq 
    \ex{n I(1+\epsilon)} p(x^i)$  $\forall x^{i}\in S_i$,
    \begin{align}
      \tr \tilde M^n_1 \rho^{\otimes n} &= \sum_{1\leq i \leq n} \sum_{x^i\in S_i} p(x^i) \nonumber\\
      &\leq  \ex{-n I(1+\epsilon)}\sum_{1\leq i \leq n} \sum_{x^i\in S_i} q(x^i) \nonumber\\
      &= \ex{-n I(1+\epsilon)} \tr \tilde M^n_1 \sigma^{\otimes n}\leq \ex{-n I(1+\epsilon)}. 
    \end{align}
    By the strong converse, this means that if $I\geq D(\sigma||\rho)$, $\lim_{n\rightarrow \infty}P_{(0)}^*=\lim_{n\rightarrow \infty}\tr \tilde M^n_1 \sigma^{\otimes n}=0$.
    
  This proves the optimality, $\bar\tau^\star \geq\bar\tau\geq (1-\epsilon)\frac{\log\bar{T}_{\text{FA}}(1+o(1))}{D(\sigma \| \rho )}$. Even more so, optimality holds also for a collection $\mathcal S$ of possible states after the change, with $I=\min_{\sigma\in\mathcal S}D(\sigma\|\rho)$.

    \textit{Change point with channels}:--- {One can define an analogous change point problem for quantum channels. In this setting, at step $n$ the algorithm receives a black-box use of a channel $\mathcal{N}^{(n)}$, which is $\mathcal{N}$ if $n\leq\nu$ and $\mathcal{M}$ if $k>\nu$, and the most general quantum strategies are allowed, such as using quantum memory, adapting operations between each use of the channel (see SM for a precise definition).} In this case, we can leverage the achievability result for states to obtain $
    \bar\tau^\star \leq \frac{\log\bar{T}_{\text{FA}}(1+o(1))}{D^{\infty}(\mathcal{M} \| \mathcal{N} )}$
    where
    $ D^\infty(\mathcal{M}\|\mathcal{N}) = \lim_{l\rightarrow\infty} \sup_\rho \frac1l D(\mathcal{M}^{\otimes l}(\rho) \| \mathcal{N}^{\otimes l}(\rho) )$ (here and in the following the input state in maximization can be any state entangled with an arbitrarily large reference system). On the other hand, we can adapt the lower bound proof using a known strong converse~\cite{fawzi2021defining}, obtaining 
    $\bar\tau^\star \geq (1-\epsilon)\frac{\log\bar{T}_{\text{FA}}(1+o(1))}{\tilde D^\infty_1(\mathcal{M}\|\mathcal{N})}$, $\forall\epsilon>0$, where $\tilde D_{\alpha}^{\infty}(\mathcal M||\mathcal N)=\lim_{l\rightarrow\infty} \sup_\rho \frac1l \tilde D_{\alpha}(\mathcal{M}^{\otimes l}(\rho) \| \mathcal{N}^{\otimes l}(\rho) )$, $\tilde D_{\alpha}(\rho||\sigma)=\frac{1}{\alpha-1}\log\tr[\sigma^{\frac{1-\alpha}{2\alpha}}\rho\sigma^{\frac{1-\alpha}{2\alpha}}]^{\alpha}$ and $\tilde{D}^\infty_1(\mathcal{M}\|\mathcal{N})=\lim_{\alpha\rightarrow 1 }\tilde D_{\alpha}^{\infty}(\mathcal M||\mathcal N)$. The quantities in the two bounds have been conjectured to coincide~\cite{Fang2021}.

    \textit{Conclusions}:--- {We have showed asymptotic optimality of the QUSUM algorithm, with a tradeoff given by the relative entropy, solving the quickest change point detection problem for quantum states in the asymptotic setting. Our results apply also to the setting when the state after the change is not known. 
%We have proposed a  measurement scheme that based on our numerical results approaches the measured relative entropy for finite block lengths,
We have proposed a measurement scheme and show numerically that it approaches the optimal tradeoff with finite block lengths. It remains unclear how to address the optimality of the quickest detection  for finite number of samples.  In the asymptotic setting, it would be interesting to find achievability results for non-i.i.d. states, especially those for which a strong converse can be found.

    \section*{Acknowledgments} 
     We thank Andreas Winter and Milán Mosonyi for useful discussions. This work was supported by the QuantERA grant C’MON-QSENS!, by Spanish MICINN PCI2019-111869-2 and 
     Spanish  Agencia Estatal de Investigación, project PID2019-107609GB-I00/AEI /10.13039/501100011033 and co-funded by the European Union Regional Development Fund within the ERDF Operational Program of Catalunya (project QuantumCat, ref. 001-P-001644).
    CH has received funding from the European Union's Horizon 2020 research and innovation programme under the Marie Sklodowska-Curie Grant Agreement No. H2020-MSCA-IF-2020-101025848. 
    JC also acknowledges support from ICREA Academia award.  MF is supported by Juan de la Cierva - Formaciòn (Spanish MICIN project FJC2021-047404-I), with funding from MCIN/AEI/10.13039/501100011033 and European Union “NextGenerationEU”/PRTR.

    \bibliography{Bib}

\begin{thebibliography}{44}%
\makeatletter
\providecommand \@ifxundefined [1]{%
 \@ifx{#1\undefined}
}%
\providecommand \@ifnum [1]{%
 \ifnum #1\expandafter \@firstoftwo
 \else \expandafter \@secondoftwo
 \fi
}%
\providecommand \@ifx [1]{%
 \ifx #1\expandafter \@firstoftwo
 \else \expandafter \@secondoftwo
 \fi
}%
\providecommand \natexlab [1]{#1}%
\providecommand \enquote  [1]{``#1''}%
\providecommand \bibnamefont  [1]{#1}%
\providecommand \bibfnamefont [1]{#1}%
\providecommand \citenamefont [1]{#1}%
\providecommand \href@noop [0]{\@secondoftwo}%
\providecommand \href [0]{\begingroup \@sanitize@url \@href}%
\providecommand \@href[1]{\@@startlink{#1}\@@href}%
\providecommand \@@href[1]{\endgroup#1\@@endlink}%
\providecommand \@sanitize@url [0]{\catcode `\\12\catcode `\$12\catcode
  `\&12\catcode `\#12\catcode `\^12\catcode `\_12\catcode `\%12\relax}%
\providecommand \@@startlink[1]{}%
\providecommand \@@endlink[0]{}%
\providecommand \url  [0]{\begingroup\@sanitize@url \@url }%
\providecommand \@url [1]{\endgroup\@href {#1}{\urlprefix }}%
\providecommand \urlprefix  [0]{URL }%
\providecommand \Eprint [0]{\href }%
\providecommand \doibase [0]{https://doi.org/}%
\providecommand \selectlanguage [0]{\@gobble}%
\providecommand \bibinfo  [0]{\@secondoftwo}%
\providecommand \bibfield  [0]{\@secondoftwo}%
\providecommand \translation [1]{[#1]}%
\providecommand \BibitemOpen [0]{}%
\providecommand \bibitemStop [0]{}%
\providecommand \bibitemNoStop [0]{.\EOS\space}%
\providecommand \EOS [0]{\spacefactor3000\relax}%
\providecommand \BibitemShut  [1]{\csname bibitem#1\endcsname}%
\let\auto@bib@innerbib\@empty
%</preamble>
\bibitem [{\citenamefont {Hawkins}\ \emph {et~al.}(2003)\citenamefont
  {Hawkins}, \citenamefont {Qiu},\ and\ \citenamefont
  {Kang}}]{hawkins2003changepoint}%
  \BibitemOpen
  \bibfield  {author} {\bibinfo {author} {\bibfnamefont {D.~M.}\ \bibnamefont
  {Hawkins}}, \bibinfo {author} {\bibfnamefont {P.}~\bibnamefont {Qiu}},\ and\
  \bibinfo {author} {\bibfnamefont {C.~W.}\ \bibnamefont {Kang}},\ }\bibfield
  {title} {\bibinfo {title} {The changepoint model for statistical process
  control},\ }\href {https://doi.org/10.1080/00224065.2003.11980233} {\bibfield
   {journal} {\bibinfo  {journal} {Journal of quality technology}\ }\textbf
  {\bibinfo {volume} {35}},\ \bibinfo {pages} {355} (\bibinfo {year}
  {2003})}\BibitemShut {NoStop}%
\bibitem [{\citenamefont {Lai}(1995)}]{lai1995sequential}%
  \BibitemOpen
  \bibfield  {author} {\bibinfo {author} {\bibfnamefont {T.~L.}\ \bibnamefont
  {Lai}},\ }\bibfield  {title} {\bibinfo {title} {Sequential changepoint
  detection in quality control and dynamical systems},\ }\href
  {https://doi.org/https://doi.org/10.1111/j.2517-6161.1995.tb02052.x}
  {\bibfield  {journal} {\bibinfo  {journal} {Journal of the Royal Statistical
  Society: Series B (Methodological)}\ }\textbf {\bibinfo {volume} {57}},\
  \bibinfo {pages} {613} (\bibinfo {year} {1995})}\BibitemShut {NoStop}%
\bibitem [{\citenamefont {Chen}\ and\ \citenamefont
  {Gupta}(2011)}]{chen2011parametric}%
  \BibitemOpen
  \bibfield  {author} {\bibinfo {author} {\bibfnamefont {J.}~\bibnamefont
  {Chen}}\ and\ \bibinfo {author} {\bibfnamefont {A.~K.}\ \bibnamefont
  {Gupta}},\ }\href@noop {} {\emph {\bibinfo {title} {Parametric statistical
  change point analysis: with applications to genetics, medicine, and
  finance}}}\ (\bibinfo  {publisher} {Springer Science \& Business Media},\
  \bibinfo {year} {2011})\BibitemShut {NoStop}%
\bibitem [{\citenamefont {Rosenfield}\ \emph {et~al.}(2010)\citenamefont
  {Rosenfield}, \citenamefont {Zhou}, \citenamefont {Wilhelm}, \citenamefont
  {Conrad}, \citenamefont {Roth},\ and\ \citenamefont
  {Meuret}}]{rosenfield2010change}%
  \BibitemOpen
  \bibfield  {author} {\bibinfo {author} {\bibfnamefont {D.}~\bibnamefont
  {Rosenfield}}, \bibinfo {author} {\bibfnamefont {E.}~\bibnamefont {Zhou}},
  \bibinfo {author} {\bibfnamefont {F.~H.}\ \bibnamefont {Wilhelm}}, \bibinfo
  {author} {\bibfnamefont {A.}~\bibnamefont {Conrad}}, \bibinfo {author}
  {\bibfnamefont {W.~T.}\ \bibnamefont {Roth}},\ and\ \bibinfo {author}
  {\bibfnamefont {A.~E.}\ \bibnamefont {Meuret}},\ }\bibfield  {title}
  {\bibinfo {title} {Change point analysis for longitudinal physiological data:
  detection of cardio-respiratory changes preceding panic attacks},\ }\href
  {https://doi.org/10.1016/j.biopsycho.2010.01.020} {\bibfield  {journal}
  {\bibinfo  {journal} {Biological psychology}\ }\textbf {\bibinfo {volume}
  {84}},\ \bibinfo {pages} {112} (\bibinfo {year} {2010})}\BibitemShut
  {NoStop}%
\bibitem [{\citenamefont {Gallagher}\ \emph {et~al.}(2013)\citenamefont
  {Gallagher}, \citenamefont {Lund},\ and\ \citenamefont
  {Robbins}}]{gallagher2013changepoint}%
  \BibitemOpen
  \bibfield  {author} {\bibinfo {author} {\bibfnamefont {C.}~\bibnamefont
  {Gallagher}}, \bibinfo {author} {\bibfnamefont {R.}~\bibnamefont {Lund}},\
  and\ \bibinfo {author} {\bibfnamefont {M.}~\bibnamefont {Robbins}},\
  }\bibfield  {title} {\bibinfo {title} {Changepoint detection in climate time
  series with long-term trends},\ }\href
  {https://doi.org/10.1175/JCLI-D-12-00704.1} {\bibfield  {journal} {\bibinfo
  {journal} {Journal of Climate}\ }\textbf {\bibinfo {volume} {26}},\ \bibinfo
  {pages} {4994} (\bibinfo {year} {2013})}\BibitemShut {NoStop}%
\bibitem [{\citenamefont {Reeves}\ \emph {et~al.}(2007)\citenamefont {Reeves},
  \citenamefont {Chen}, \citenamefont {Wang}, \citenamefont {Lund},\ and\
  \citenamefont {Lu}}]{reeves2007review}%
  \BibitemOpen
  \bibfield  {author} {\bibinfo {author} {\bibfnamefont {J.}~\bibnamefont
  {Reeves}}, \bibinfo {author} {\bibfnamefont {J.}~\bibnamefont {Chen}},
  \bibinfo {author} {\bibfnamefont {X.~L.}\ \bibnamefont {Wang}}, \bibinfo
  {author} {\bibfnamefont {R.}~\bibnamefont {Lund}},\ and\ \bibinfo {author}
  {\bibfnamefont {Q.~Q.}\ \bibnamefont {Lu}},\ }\bibfield  {title} {\bibinfo
  {title} {A review and comparison of changepoint detection techniques for
  climate data},\ }\href {https://doi.org/https://doi.org/10.1175/JAM2493.1}
  {\bibfield  {journal} {\bibinfo  {journal} {Journal of applied meteorology
  and climatology}\ }\textbf {\bibinfo {volume} {46}},\ \bibinfo {pages} {900}
  (\bibinfo {year} {2007})}\BibitemShut {NoStop}%
\bibitem [{\citenamefont {Kurt}\ \emph {et~al.}(2018)\citenamefont {Kurt},
  \citenamefont {Yılmaz},\ and\ \citenamefont {Wang}}]{kurt_distributed_2018}%
  \BibitemOpen
  \bibfield  {author} {\bibinfo {author} {\bibfnamefont {M.~N.}\ \bibnamefont
  {Kurt}}, \bibinfo {author} {\bibfnamefont {Y.}~\bibnamefont {Yılmaz}},\ and\
  \bibinfo {author} {\bibfnamefont {X.}~\bibnamefont {Wang}},\ }\bibfield
  {title} {\bibinfo {title} {Distributed {Quickest} {Detection} of
  {Cyber}-{Attacks} in {Smart} {Grid}},\ }\href
  {https://doi.org/10.1109/TIFS.2018.2800908} {\bibfield  {journal} {\bibinfo
  {journal} {IEEE Transactions on Information Forensics and Security}\ }\textbf
  {\bibinfo {volume} {13}},\ \bibinfo {pages} {2015} (\bibinfo {year}
  {2018})}\BibitemShut {NoStop}%
\bibitem [{\citenamefont {Niekum}\ \emph {et~al.}(2015)\citenamefont {Niekum},
  \citenamefont {Osentoski}, \citenamefont {Atkeson},\ and\ \citenamefont
  {Barto}}]{niekum2015online}%
  \BibitemOpen
  \bibfield  {author} {\bibinfo {author} {\bibfnamefont {S.}~\bibnamefont
  {Niekum}}, \bibinfo {author} {\bibfnamefont {S.}~\bibnamefont {Osentoski}},
  \bibinfo {author} {\bibfnamefont {C.~G.}\ \bibnamefont {Atkeson}},\ and\
  \bibinfo {author} {\bibfnamefont {A.~G.}\ \bibnamefont {Barto}},\ }\bibfield
  {title} {\bibinfo {title} {Online bayesian changepoint detection for
  articulated motion models},\ }in\ \href
  {https://doi.org/10.1109/ICRA.2015.7139383} {\emph {\bibinfo {booktitle}
  {2015 IEEE International Conference on Robotics and Automation (ICRA)}}}\
  (\bibinfo {organization} {IEEE},\ \bibinfo {year} {2015})\ pp.\ \bibinfo
  {pages} {1468--1475}\BibitemShut {NoStop}%
\bibitem [{\citenamefont {Ranganathan}(2012)}]{ranganathan2012pliss}%
  \BibitemOpen
  \bibfield  {author} {\bibinfo {author} {\bibfnamefont {A.}~\bibnamefont
  {Ranganathan}},\ }\bibfield  {title} {\bibinfo {title} {Pliss: labeling
  places using online changepoint detection},\ }\href
  {https://doi.org/https://doi.org/10.1007/s10514-012-9273-4} {\bibfield
  {journal} {\bibinfo  {journal} {Autonomous Robots}\ }\textbf {\bibinfo
  {volume} {32}},\ \bibinfo {pages} {351} (\bibinfo {year} {2012})}\BibitemShut
  {NoStop}%
\bibitem [{\citenamefont {Basseville}\ \emph {et~al.}(1993)\citenamefont
  {Basseville}, \citenamefont {Nikiforov} \emph {et~al.}}]{basseville1993}%
  \BibitemOpen
  \bibfield  {author} {\bibinfo {author} {\bibfnamefont {M.}~\bibnamefont
  {Basseville}}, \bibinfo {author} {\bibfnamefont {I.~V.}\ \bibnamefont
  {Nikiforov}}, \emph {et~al.},\ }\href@noop {} {\emph {\bibinfo {title}
  {Detection of abrupt changes: theory and application}}},\ Vol.\ \bibinfo
  {volume} {104}\ (\bibinfo  {publisher} {Prentice Hall Englewood Cliffs},\
  \bibinfo {year} {1993})\BibitemShut {NoStop}%
\bibitem [{\citenamefont {Poor}\ and\ \citenamefont
  {Hadjiliadis}(2008)}]{poor2008}%
  \BibitemOpen
  \bibfield  {author} {\bibinfo {author} {\bibfnamefont {H.~V.}\ \bibnamefont
  {Poor}}\ and\ \bibinfo {author} {\bibfnamefont {O.}~\bibnamefont
  {Hadjiliadis}},\ }\href@noop {} {\emph {\bibinfo {title} {Quickest
  detection}}}\ (\bibinfo  {publisher} {Cambridge University Press},\ \bibinfo
  {year} {2008})\BibitemShut {NoStop}%
\bibitem [{\citenamefont {Tartakovsky}\ \emph {et~al.}(2014)\citenamefont
  {Tartakovsky}, \citenamefont {Nikiforov},\ and\ \citenamefont
  {Basseville}}]{tartakovsky2014}%
  \BibitemOpen
  \bibfield  {author} {\bibinfo {author} {\bibfnamefont {A.}~\bibnamefont
  {Tartakovsky}}, \bibinfo {author} {\bibfnamefont {I.}~\bibnamefont
  {Nikiforov}},\ and\ \bibinfo {author} {\bibfnamefont {M.}~\bibnamefont
  {Basseville}},\ }\href@noop {} {\emph {\bibinfo {title} {Sequential analysis:
  Hypothesis testing and changepoint detection}}}\ (\bibinfo  {publisher}
  {Chapman and Hall/CRC},\ \bibinfo {year} {2014})\BibitemShut {NoStop}%
\bibitem [{\citenamefont {Akimoto}\ and\ \citenamefont
  {Hayashi}(2011)}]{Akimoto2011}%
  \BibitemOpen
  \bibfield  {author} {\bibinfo {author} {\bibfnamefont {D.}~\bibnamefont
  {Akimoto}}\ and\ \bibinfo {author} {\bibfnamefont {M.}~\bibnamefont
  {Hayashi}},\ }\bibfield  {title} {\bibinfo {title} {Discrimination of the
  change point in a quantum setting},\ }\href
  {https://doi.org/10.1103/PhysRevA.83.052328} {\bibfield  {journal} {\bibinfo
  {journal} {Phys. Rev. A}\ }\textbf {\bibinfo {volume} {83}},\ \bibinfo
  {pages} {052328} (\bibinfo {year} {2011})}\BibitemShut {NoStop}%
\bibitem [{\citenamefont {Sent{\'\i}s}\ \emph {et~al.}(2016)\citenamefont
  {Sent{\'\i}s}, \citenamefont {Bagan}, \citenamefont {Calsamiglia},
  \citenamefont {Chiribella},\ and\ \citenamefont {Munoz-Tapia}}]{QCP1}%
  \BibitemOpen
  \bibfield  {author} {\bibinfo {author} {\bibfnamefont {G.}~\bibnamefont
  {Sent{\'\i}s}}, \bibinfo {author} {\bibfnamefont {E.}~\bibnamefont {Bagan}},
  \bibinfo {author} {\bibfnamefont {J.}~\bibnamefont {Calsamiglia}}, \bibinfo
  {author} {\bibfnamefont {G.}~\bibnamefont {Chiribella}},\ and\ \bibinfo
  {author} {\bibfnamefont {R.}~\bibnamefont {Munoz-Tapia}},\ }\bibfield
  {title} {\bibinfo {title} {Quantum change point},\ }\href
  {https://doi.org/https://doi.org/10.1103/PhysRevLett.117.150502} {\bibfield
  {journal} {\bibinfo  {journal} {Physical Review Letters}\ }\textbf {\bibinfo
  {volume} {117}},\ \bibinfo {pages} {150502} (\bibinfo {year}
  {2016})}\BibitemShut {NoStop}%
\bibitem [{\citenamefont {Sent{\'\i}s}\ \emph {et~al.}(2017)\citenamefont
  {Sent{\'\i}s}, \citenamefont {Calsamiglia},\ and\ \citenamefont
  {Mu{\~n}oz-Tapia}}]{QCP2}%
  \BibitemOpen
  \bibfield  {author} {\bibinfo {author} {\bibfnamefont {G.}~\bibnamefont
  {Sent{\'\i}s}}, \bibinfo {author} {\bibfnamefont {J.}~\bibnamefont
  {Calsamiglia}},\ and\ \bibinfo {author} {\bibfnamefont {R.}~\bibnamefont
  {Mu{\~n}oz-Tapia}},\ }\bibfield  {title} {\bibinfo {title} {Exact
  identification of a quantum change point},\ }\href
  {https://doi.org/https://doi.org/10.1103/PhysRevLett.119.140506} {\bibfield
  {journal} {\bibinfo  {journal} {Physical Review Letters}\ }\textbf {\bibinfo
  {volume} {119}},\ \bibinfo {pages} {140506} (\bibinfo {year}
  {2017})}\BibitemShut {NoStop}%
\bibitem [{\citenamefont {Sent{\'\i}s}\ \emph {et~al.}(2018)\citenamefont
  {Sent{\'\i}s}, \citenamefont {Mart{\'\i}nez-Vargas},\ and\ \citenamefont
  {Mu{\~n}oz-Tapia}}]{QCPonline}%
  \BibitemOpen
  \bibfield  {author} {\bibinfo {author} {\bibfnamefont {G.}~\bibnamefont
  {Sent{\'\i}s}}, \bibinfo {author} {\bibfnamefont {E.}~\bibnamefont
  {Mart{\'\i}nez-Vargas}},\ and\ \bibinfo {author} {\bibfnamefont
  {R.}~\bibnamefont {Mu{\~n}oz-Tapia}},\ }\bibfield  {title} {\bibinfo {title}
  {Online strategies for exactly identifying a quantum change point},\ }\href
  {https://doi.org/https://doi.org/10.1103/PhysRevA.98.052305} {\bibfield
  {journal} {\bibinfo  {journal} {Physical Review A}\ }\textbf {\bibinfo
  {volume} {98}},\ \bibinfo {pages} {052305} (\bibinfo {year}
  {2018})}\BibitemShut {NoStop}%
\bibitem [{\citenamefont {Page}(1954)}]{page54}%
  \BibitemOpen
  \bibfield  {author} {\bibinfo {author} {\bibfnamefont {E.~S.}\ \bibnamefont
  {Page}},\ }\bibfield  {title} {\bibinfo {title} {Continuous inspection
  schemes},\ }\href {https://doi.org/https://doi.org/10.2307/2333009}
  {\bibfield  {journal} {\bibinfo  {journal} {Biometrika}\ }\textbf {\bibinfo
  {volume} {41}},\ \bibinfo {pages} {100} (\bibinfo {year} {1954})}\BibitemShut
  {NoStop}%
\bibitem [{\citenamefont {Lorden}\ \emph {et~al.}(1971)\citenamefont {Lorden}
  \emph {et~al.}}]{lorden71}%
  \BibitemOpen
  \bibfield  {author} {\bibinfo {author} {\bibfnamefont {G.}~\bibnamefont
  {Lorden}} \emph {et~al.},\ }\bibfield  {title} {\bibinfo {title} {Procedures
  for reacting to a change in distribution},\ }\href
  {https://doi.org/https://doi.org/10.1214/aoms/1177693055} {\bibfield
  {journal} {\bibinfo  {journal} {The Annals of Mathematical Statistics}\
  }\textbf {\bibinfo {volume} {42}},\ \bibinfo {pages} {1897} (\bibinfo {year}
  {1971})}\BibitemShut {NoStop}%
\bibitem [{\citenamefont {Moustakides}\ \emph {et~al.}(1986)\citenamefont
  {Moustakides} \emph {et~al.}}]{moustakides1986optimal}%
  \BibitemOpen
  \bibfield  {author} {\bibinfo {author} {\bibfnamefont {G.~V.}\ \bibnamefont
  {Moustakides}} \emph {et~al.},\ }\bibfield  {title} {\bibinfo {title}
  {Optimal stopping times for detecting changes in distributions},\ }\href
  {https://doi.org/10.1214/aos/1176350164} {\bibfield  {journal} {\bibinfo
  {journal} {The Annals of Statistics}\ }\textbf {\bibinfo {volume} {14}},\
  \bibinfo {pages} {1379} (\bibinfo {year} {1986})}\BibitemShut {NoStop}%
\bibitem [{\citenamefont {Ritov}(1990)}]{ritov1990decision}%
  \BibitemOpen
  \bibfield  {author} {\bibinfo {author} {\bibfnamefont {Y.}~\bibnamefont
  {Ritov}},\ }\bibfield  {title} {\bibinfo {title} {Decision theoretic
  optimality of the cusum procedure},\ }\href
  {https://doi.org/10.1214/aos/1176347761} {\bibfield  {journal} {\bibinfo
  {journal} {The Annals of Statistics}\ ,\ \bibinfo {pages} {1464}} (\bibinfo
  {year} {1990})}\BibitemShut {NoStop}%
\bibitem [{\citenamefont {Vargas}\ \emph {et~al.}(2021)\citenamefont {Vargas},
  \citenamefont {Hirche}, \citenamefont {Sent{\'\i}s}, \citenamefont
  {Skotiniotis}, \citenamefont {Carrizo}, \citenamefont {Mu{\~n}oz-Tapia},\
  and\ \citenamefont {Calsamiglia}}]{vargas2021quantum}%
  \BibitemOpen
  \bibfield  {author} {\bibinfo {author} {\bibfnamefont {E.~M.}\ \bibnamefont
  {Vargas}}, \bibinfo {author} {\bibfnamefont {C.}~\bibnamefont {Hirche}},
  \bibinfo {author} {\bibfnamefont {G.}~\bibnamefont {Sent{\'\i}s}}, \bibinfo
  {author} {\bibfnamefont {M.}~\bibnamefont {Skotiniotis}}, \bibinfo {author}
  {\bibfnamefont {M.}~\bibnamefont {Carrizo}}, \bibinfo {author} {\bibfnamefont
  {R.}~\bibnamefont {Mu{\~n}oz-Tapia}},\ and\ \bibinfo {author} {\bibfnamefont
  {J.}~\bibnamefont {Calsamiglia}},\ }\bibfield  {title} {\bibinfo {title}
  {Quantum sequential hypothesis testing},\ }\href
  {https://doi.org/https://doi.org/10.1103/PhysRevLett.126.180502} {\bibfield
  {journal} {\bibinfo  {journal} {Physical Review Letters}\ }\textbf {\bibinfo
  {volume} {126}},\ \bibinfo {pages} {180502} (\bibinfo {year}
  {2021})}\BibitemShut {NoStop}%
\bibitem [{\citenamefont {Li}\ \emph {et~al.}(2022{\natexlab{a}})\citenamefont
  {Li}, \citenamefont {Tan},\ and\ \citenamefont {Tomamichel}}]{li2022optimal}%
  \BibitemOpen
  \bibfield  {author} {\bibinfo {author} {\bibfnamefont {Y.}~\bibnamefont
  {Li}}, \bibinfo {author} {\bibfnamefont {V.~Y.}\ \bibnamefont {Tan}},\ and\
  \bibinfo {author} {\bibfnamefont {M.}~\bibnamefont {Tomamichel}},\ }\bibfield
   {title} {\bibinfo {title} {Optimal adaptive strategies for sequential
  quantum hypothesis testing},\ }\href
  {https://doi.org/https://doi.org/10.1007/s00220-022-04362-5} {\bibfield
  {journal} {\bibinfo  {journal} {Communications in Mathematical Physics}\
  }\textbf {\bibinfo {volume} {392}},\ \bibinfo {pages} {993} (\bibinfo {year}
  {2022}{\natexlab{a}})}\BibitemShut {NoStop}%
\bibitem [{\citenamefont {Li}\ \emph {et~al.}(2022{\natexlab{b}})\citenamefont
  {Li}, \citenamefont {Hirche},\ and\ \citenamefont {Tomamichel}}]{SQCD}%
  \BibitemOpen
  \bibfield  {author} {\bibinfo {author} {\bibfnamefont {Y.}~\bibnamefont
  {Li}}, \bibinfo {author} {\bibfnamefont {C.}~\bibnamefont {Hirche}},\ and\
  \bibinfo {author} {\bibfnamefont {M.}~\bibnamefont {Tomamichel}},\ }\bibfield
   {title} {\bibinfo {title} {Sequential quantum channel discrimination},\ }in\
  \href@noop {} {\emph {\bibinfo {booktitle} {2022 IEEE International Symposium
  on Information Theory (ISIT)}}}\ (\bibinfo {year} {2022})\BibitemShut
  {NoStop}%
\bibitem [{Note1()}]{Note1}%
  \BibitemOpen
  \bibinfo {note} {If POVMs with non-discrete outcomes are allowed, then one
  may consider also $\protect \mathaccentV {bar}016\tau ^{\star } :=\protect
  \qopname \relax m{sup}_{\nu \geq 0} \protect \mathrm {ess\protect \tmspace
  +\thinmuskip {.1667em} sup}\protect \tmspace +\thinmuskip {.1667em}\protect
  \mathbb {E}_{\nu }[T-\nu | T > \nu , X_1,...,X_{\nu }]$, where the essential
  supremum means that the optimization is done a supremum over set of
  measurement outcomes of non-zero measure, where the measure of a set is given
  by its probability. For simplicity, in this paper we consider only
  measurement with discrete outcomes}\BibitemShut {NoStop}%
\bibitem [{\citenamefont {Wald}(2004)}]{wald2004sequential}%
  \BibitemOpen
  \bibfield  {author} {\bibinfo {author} {\bibfnamefont {A.}~\bibnamefont
  {Wald}},\ }\href@noop {} {\emph {\bibinfo {title} {Sequential analysis}}}\
  (\bibinfo  {publisher} {Courier Corporation},\ \bibinfo {year}
  {2004})\BibitemShut {NoStop}%
\bibitem [{\citenamefont {Petz}(1988)}]{petz88}%
  \BibitemOpen
  \bibfield  {author} {\bibinfo {author} {\bibfnamefont {D.}~\bibnamefont
  {Petz}},\ }\bibfield  {title} {\bibinfo {title} {A variational expression for
  the relative entropy},\ }\href {https://doi.org/doi:10.1007/BF01225040}
  {\bibfield  {journal} {\bibinfo  {journal} {Communications in Mathematical
  Physics}\ }\textbf {\bibinfo {volume} {114}},\ \bibinfo {pages} {345}
  (\bibinfo {year} {1988})}\BibitemShut {NoStop}%
\bibitem [{\citenamefont {Hayashi}(2001)}]{Hayashi_2001}%
  \BibitemOpen
  \bibfield  {author} {\bibinfo {author} {\bibfnamefont {M.}~\bibnamefont
  {Hayashi}},\ }\bibfield  {title} {\bibinfo {title} {Asymptotics of quantum
  relative entropy from a representation theoretical viewpoint},\ }\href
  {https://doi.org/10.1088/0305-4470/34/16/309} {\bibfield  {journal} {\bibinfo
   {journal} {Journal of Physics A: Mathematical and General}\ }\textbf
  {\bibinfo {volume} {34}},\ \bibinfo {pages} {3413} (\bibinfo {year}
  {2001})}\BibitemShut {NoStop}%
\bibitem [{\citenamefont {Berta}\ \emph {et~al.}(2017)\citenamefont {Berta},
  \citenamefont {Fawzi},\ and\ \citenamefont
  {Tomamichel}}]{berta2017variational}%
  \BibitemOpen
  \bibfield  {author} {\bibinfo {author} {\bibfnamefont {M.}~\bibnamefont
  {Berta}}, \bibinfo {author} {\bibfnamefont {O.}~\bibnamefont {Fawzi}},\ and\
  \bibinfo {author} {\bibfnamefont {M.}~\bibnamefont {Tomamichel}},\ }\bibfield
   {title} {\bibinfo {title} {On variational expressions for quantum relative
  entropies},\ }\href
  {https://doi.org/https://doi.org/10.1007/s11005-017-0990-7} {\bibfield
  {journal} {\bibinfo  {journal} {Letters in Mathematical Physics}\ }\textbf
  {\bibinfo {volume} {107}},\ \bibinfo {pages} {2239} (\bibinfo {year}
  {2017})}\BibitemShut {NoStop}%
\bibitem [{Note2()}]{Note2}%
  \BibitemOpen
  \bibinfo {note} {Notice also that an upper bound on the sufficient size of
  the blocks to get the achievability in the theorem is $k=\protect
  \mathaccentV {tilde}07EO(d/(\epsilon D(\sigma ||\rho ))$ (hiding logarithmic
  factors), with a better dependence on the dimension than optimal tomography,
  which requires $\Theta (d^2)$ copies~\cite {haah2017sample, o2016efficient}.
  In the SM we argue that for any fixed pair of states one can in fact choose
  $l$ depending only on relative entropies of the pair, therefore not
  explicitly on the dimension, based on~\cite
  {audenaert2012quantum}.}\BibitemShut {Stop}%
\bibitem [{\citenamefont {Fawzi}\ \emph {et~al.}(2019)\citenamefont {Fawzi},
  \citenamefont {Saunderson},\ and\ \citenamefont {Parrilo}}]{cvxquad}%
  \BibitemOpen
  \bibfield  {author} {\bibinfo {author} {\bibfnamefont {H.}~\bibnamefont
  {Fawzi}}, \bibinfo {author} {\bibfnamefont {J.}~\bibnamefont {Saunderson}},\
  and\ \bibinfo {author} {\bibfnamefont {P.~A.}\ \bibnamefont {Parrilo}},\
  }\bibfield  {title} {\bibinfo {title} {Semidefinite approximations of the
  matrix logarithm},\ }\href {https://doi.org/10.1007/s10208-018-9385-0}
  {\bibfield  {journal} {\bibinfo  {journal} {Foundations of Computational
  Mathematics}\ }\textbf {\bibinfo {volume} {19}},\ \bibinfo {pages} {259}
  (\bibinfo {year} {2019})}\BibitemShut {NoStop}%
\bibitem [{\citenamefont {Bacon}\ \emph {et~al.}(2006)\citenamefont {Bacon},
  \citenamefont {Chuang},\ and\ \citenamefont {Harrow}}]{bacon_efficient_2006}%
  \BibitemOpen
  \bibfield  {author} {\bibinfo {author} {\bibfnamefont {D.}~\bibnamefont
  {Bacon}}, \bibinfo {author} {\bibfnamefont {I.~L.}\ \bibnamefont {Chuang}},\
  and\ \bibinfo {author} {\bibfnamefont {A.~W.}\ \bibnamefont {Harrow}},\
  }\bibfield  {title} {\bibinfo {title} {Efficient {Quantum} {Circuits} for
  {Schur} and {Clebsch}-{Gordan} {Transforms}},\ }\href
  {https://doi.org/10.1103/PhysRevLett.97.170502} {\bibfield  {journal}
  {\bibinfo  {journal} {Physical Review Letters}\ }\textbf {\bibinfo {volume}
  {97}},\ \bibinfo {pages} {170502} (\bibinfo {year} {2006})}\BibitemShut
  {NoStop}%
\bibitem [{\citenamefont {Harrow}(2005)}]{harrow_phd}%
  \BibitemOpen
  \bibfield  {author} {\bibinfo {author} {\bibfnamefont {A.~W.}\ \bibnamefont
  {Harrow}},\ }\bibfield  {title} {\bibinfo {title} {{Applications of coherent
  classical communication and Schur duality to quantum information theory}},\
  }\href@noop {} {\bibfield  {journal} {\bibinfo  {journal} {PhD thesis, MIT}\
  } (\bibinfo {year} {2005})}\BibitemShut {NoStop}%
\bibitem [{\citenamefont {Krovi}(2019)}]{krovi_efficient_2019}%
  \BibitemOpen
  \bibfield  {author} {\bibinfo {author} {\bibfnamefont {H.}~\bibnamefont
  {Krovi}},\ }\bibfield  {title} {\bibinfo {title} {An efficient high
  dimensional quantum schur transform},\ }\href
  {https://doi.org/10.22331/q-2019-02-14-122} {\bibfield  {journal} {\bibinfo
  {journal} {Quantum}\ }\textbf {\bibinfo {volume} {3}},\ \bibinfo {pages}
  {122} (\bibinfo {year} {2019})}\BibitemShut {NoStop}%
\bibitem [{\citenamefont {Lai}(1998)}]{lai98}%
  \BibitemOpen
  \bibfield  {author} {\bibinfo {author} {\bibfnamefont {T.~L.}\ \bibnamefont
  {Lai}},\ }\bibfield  {title} {\bibinfo {title} {Information bounds and quick
  detection of parameter changes in stochastic systems},\ }\href
  {https://doi.org/10.1109/18.737522} {\bibfield  {journal} {\bibinfo
  {journal} {IEEE Transactions on Information Theory}\ }\textbf {\bibinfo
  {volume} {44}},\ \bibinfo {pages} {2917} (\bibinfo {year}
  {1998})}\BibitemShut {NoStop}%
\bibitem [{Note3()}]{Note3}%
  \BibitemOpen
  \bibinfo {note} {Ref.~\cite {lai98} shows that the proof applies also when
  the probability in presence of a change point, conditioned on past events,
  can depend on the change point location. The statement of the theorem
  in~\cite {lai98} has an extra condition about the convergence of the
  log-likelihood which is not needed for our purposes and expresses the
  tradeoff only as a limit for large $\protect \qopname \relax o{log}\protect
  \mathaccentV {bar}016{T}_{\protect \text {FA}}$ and $\epsilon \rightarrow 0$,
  while we state it with finite $\epsilon $}\BibitemShut {NoStop}%
\bibitem [{\citenamefont {Ogawa}\ and\ \citenamefont
  {Nagaoka}(2000)}]{Ogawa2000}%
  \BibitemOpen
  \bibfield  {author} {\bibinfo {author} {\bibfnamefont {T.}~\bibnamefont
  {Ogawa}}\ and\ \bibinfo {author} {\bibfnamefont {H.}~\bibnamefont
  {Nagaoka}},\ }\bibfield  {title} {\bibinfo {title} {Strong converse and
  stein's lemma in quantum hypothesis testing},\ }\href
  {https://doi.org/10.1109/18.887855} {\bibfield  {journal} {\bibinfo
  {journal} {IEEE Transactions on Information Theory}\ }\textbf {\bibinfo
  {volume} {46}},\ \bibinfo {pages} {2428} (\bibinfo {year}
  {2000})}\BibitemShut {NoStop}%
\bibitem [{\citenamefont {Fawzi}\ and\ \citenamefont
  {Fawzi}(2021)}]{fawzi2021defining}%
  \BibitemOpen
  \bibfield  {author} {\bibinfo {author} {\bibfnamefont {H.}~\bibnamefont
  {Fawzi}}\ and\ \bibinfo {author} {\bibfnamefont {O.}~\bibnamefont {Fawzi}},\
  }\bibfield  {title} {\bibinfo {title} {Defining quantum divergences via
  convex optimization},\ }\href
  {https://doi.org/https://doi.org/10.22331/q-2021-01-26-387} {\bibfield
  {journal} {\bibinfo  {journal} {Quantum}\ }\textbf {\bibinfo {volume} {5}},\
  \bibinfo {pages} {387} (\bibinfo {year} {2021})}\BibitemShut {NoStop}%
\bibitem [{\citenamefont {Fang}\ \emph {et~al.}(2021)\citenamefont {Fang},
  \citenamefont {Gour},\ and\ \citenamefont {Wang}}]{Fang2021}%
  \BibitemOpen
  \bibfield  {author} {\bibinfo {author} {\bibfnamefont {K.}~\bibnamefont
  {Fang}}, \bibinfo {author} {\bibfnamefont {G.}~\bibnamefont {Gour}},\ and\
  \bibinfo {author} {\bibfnamefont {X.}~\bibnamefont {Wang}},\ }\href
  {https://doi.org/10.48550/ARXIV.2110.14842} {\bibinfo {title} {Towards the
  ultimate limits of quantum channel discrimination}} (\bibinfo {year}
  {2021})\BibitemShut {NoStop}%
\bibitem [{\citenamefont {Haah}\ \emph {et~al.}(2017)\citenamefont {Haah},
  \citenamefont {Harrow}, \citenamefont {Ji}, \citenamefont {Wu},\ and\
  \citenamefont {Yu}}]{haah2017sample}%
  \BibitemOpen
  \bibfield  {author} {\bibinfo {author} {\bibfnamefont {J.}~\bibnamefont
  {Haah}}, \bibinfo {author} {\bibfnamefont {A.~W.}\ \bibnamefont {Harrow}},
  \bibinfo {author} {\bibfnamefont {Z.}~\bibnamefont {Ji}}, \bibinfo {author}
  {\bibfnamefont {X.}~\bibnamefont {Wu}},\ and\ \bibinfo {author}
  {\bibfnamefont {N.}~\bibnamefont {Yu}},\ }\bibfield  {title} {\bibinfo
  {title} {Sample-optimal tomography of quantum states},\ }\href
  {https://doi.org/10.1109/TIT.2017.2719044} {\bibfield  {journal} {\bibinfo
  {journal} {IEEE Transactions on Information Theory}\ }\textbf {\bibinfo
  {volume} {63}},\ \bibinfo {pages} {5628} (\bibinfo {year}
  {2017})}\BibitemShut {NoStop}%
\bibitem [{\citenamefont {O'Donnell}\ and\ \citenamefont
  {Wright}(2016)}]{o2016efficient}%
  \BibitemOpen
  \bibfield  {author} {\bibinfo {author} {\bibfnamefont {R.}~\bibnamefont
  {O'Donnell}}\ and\ \bibinfo {author} {\bibfnamefont {J.}~\bibnamefont
  {Wright}},\ }\bibfield  {title} {\bibinfo {title} {Efficient quantum
  tomography},\ }in\ \href
  {https://doi.org/https://doi.org/10.48550/arXiv.1508.01907} {\emph {\bibinfo
  {booktitle} {Proceedings of the forty-eighth annual ACM symposium on Theory
  of Computing}}}\ (\bibinfo {year} {2016})\ pp.\ \bibinfo {pages}
  {899--912}\BibitemShut {NoStop}%
\bibitem [{\citenamefont {Audenaert}\ \emph {et~al.}(2012)\citenamefont
  {Audenaert}, \citenamefont {Mosonyi},\ and\ \citenamefont
  {Verstraete}}]{audenaert2012quantum}%
  \BibitemOpen
  \bibfield  {author} {\bibinfo {author} {\bibfnamefont {K.~M.}\ \bibnamefont
  {Audenaert}}, \bibinfo {author} {\bibfnamefont {M.}~\bibnamefont {Mosonyi}},\
  and\ \bibinfo {author} {\bibfnamefont {F.}~\bibnamefont {Verstraete}},\
  }\bibfield  {title} {\bibinfo {title} {Quantum state discrimination bounds
  for finite sample size},\ }\href@noop {} {\bibfield  {journal} {\bibinfo
  {journal} {Journal of Mathematical Physics}\ }\textbf {\bibinfo {volume}
  {53}},\ \bibinfo {pages} {122205} (\bibinfo {year} {2012})}\BibitemShut
  {NoStop}%
\bibitem [{\citenamefont {Wilde}\ \emph {et~al.}(2020)\citenamefont {Wilde},
  \citenamefont {Berta}, \citenamefont {Hirche},\ and\ \citenamefont
  {Kaur}}]{wilde2020amortized}%
  \BibitemOpen
  \bibfield  {author} {\bibinfo {author} {\bibfnamefont {M.~M.}\ \bibnamefont
  {Wilde}}, \bibinfo {author} {\bibfnamefont {M.}~\bibnamefont {Berta}},
  \bibinfo {author} {\bibfnamefont {C.}~\bibnamefont {Hirche}},\ and\ \bibinfo
  {author} {\bibfnamefont {E.}~\bibnamefont {Kaur}},\ }\bibfield  {title}
  {\bibinfo {title} {Amortized channel divergence for asymptotic quantum
  channel discrimination},\ }\href
  {https://doi.org/https://doi.org/10.1007/s11005-020-01297-7} {\bibfield
  {journal} {\bibinfo  {journal} {Letters in Mathematical Physics}\ }\textbf
  {\bibinfo {volume} {110}},\ \bibinfo {pages} {2277} (\bibinfo {year}
  {2020})}\BibitemShut {NoStop}%
\bibitem [{\citenamefont {Wang}\ and\ \citenamefont
  {Wilde}(2019)}]{wang2019resource}%
  \BibitemOpen
  \bibfield  {author} {\bibinfo {author} {\bibfnamefont {X.}~\bibnamefont
  {Wang}}\ and\ \bibinfo {author} {\bibfnamefont {M.~M.}\ \bibnamefont
  {Wilde}},\ }\bibfield  {title} {\bibinfo {title} {Resource theory of
  asymmetric distinguishability for quantum channels},\ }\href
  {https://doi.org/https://doi.org/10.1103/PhysRevResearch.1.033169} {\bibfield
   {journal} {\bibinfo  {journal} {Physical Review Research}\ }\textbf
  {\bibinfo {volume} {1}},\ \bibinfo {pages} {033169} (\bibinfo {year}
  {2019})}\BibitemShut {NoStop}%
\bibitem [{\citenamefont {Fang}\ \emph {et~al.}(2020)\citenamefont {Fang},
  \citenamefont {Fawzi}, \citenamefont {Renner},\ and\ \citenamefont
  {Sutter}}]{fang2020chain}%
  \BibitemOpen
  \bibfield  {author} {\bibinfo {author} {\bibfnamefont {K.}~\bibnamefont
  {Fang}}, \bibinfo {author} {\bibfnamefont {O.}~\bibnamefont {Fawzi}},
  \bibinfo {author} {\bibfnamefont {R.}~\bibnamefont {Renner}},\ and\ \bibinfo
  {author} {\bibfnamefont {D.}~\bibnamefont {Sutter}},\ }\bibfield  {title}
  {\bibinfo {title} {Chain rule for the quantum relative entropy},\ }\href
  {https://doi.org/10.1103/PhysRevLett.124.100501} {\bibfield  {journal}
  {\bibinfo  {journal} {Physical Review Letters}\ }\textbf {\bibinfo {volume}
  {124}},\ \bibinfo {pages} {100501} (\bibinfo {year} {2020})}\BibitemShut
  {NoStop}%
\end{thebibliography}%
    %
    %\clearpage
    %\setcounter{equation}{0}
    %\begin{widetext}
    %\appendix*
    %\setcounter{page}{1}

    %\clearpage
    %\end{widetext}

\clearpage

\newpage

\onecolumngrid
\setlength{\baselineskip}{20pt}

% \begin{widetext}
\appendix*
\setcounter{equation}{0}
\renewcommand{\theequation}{S\arabic{equation}}
%\renewcommand{\theequation}{\arabic{equation}}
% \onecolumngrid

\section{SUPPLEMENTAL MATERIAL}

    \subsection{Graphical aid to classical CUSUM for change-point detection}
    \begin{figure}[ht]
    \includegraphics[scale=.25]{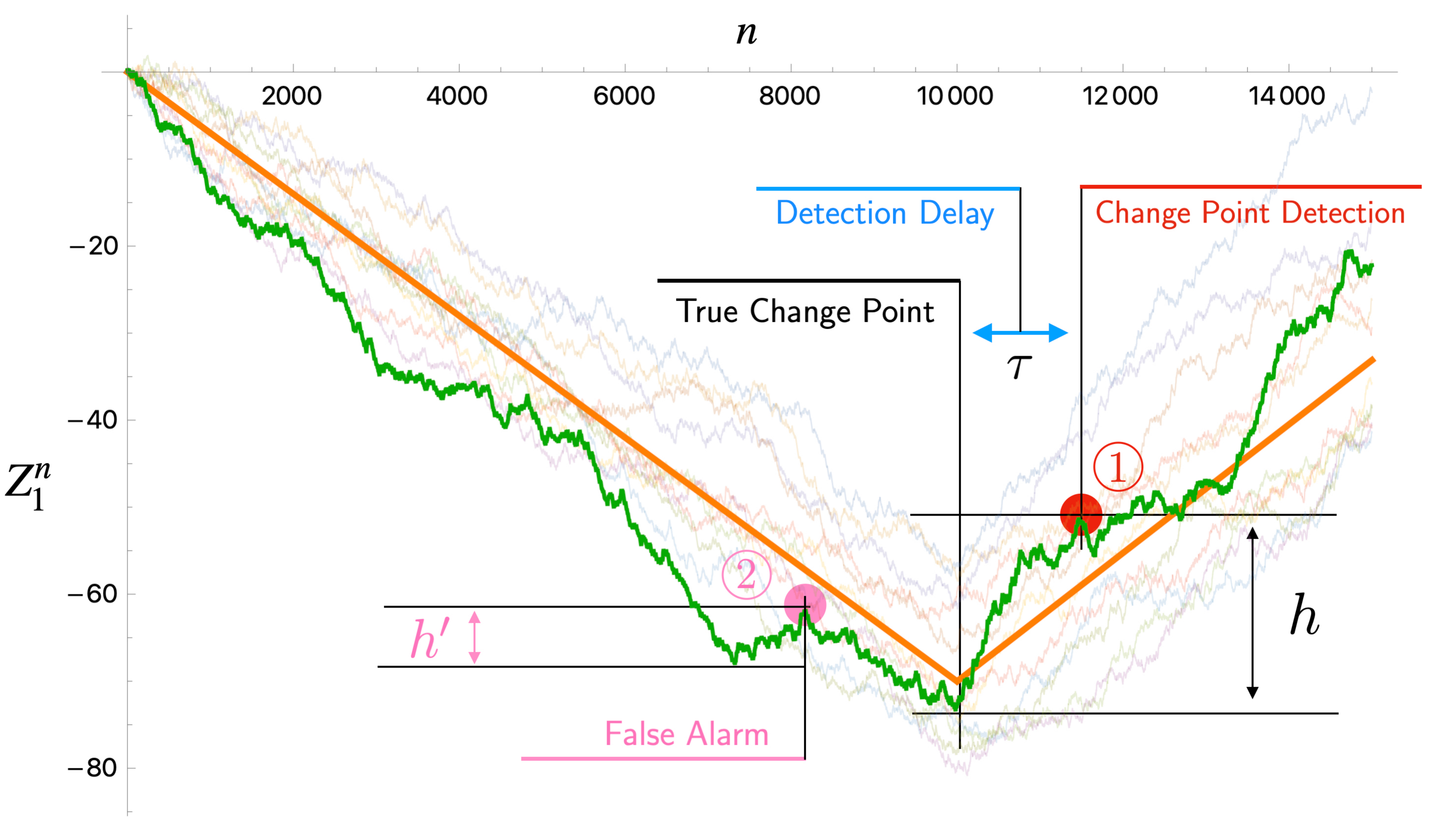}
    %\vspace{-1em}
    %\caption{Graphical representation of the CUSUM algorithm.}
    \label{fig:1S}
    \end{figure}
    This figure illustrates how the CUSUM algorithm works in a (classical) setting where samples $x_1^{n}=\{x_1,x_2,\ldots, x_n\}$ are obtained from a Bernoulli trial (coin toss) with a bias $p=1/5$ that at time $n=\nu=10^4$ changes to bias $q=1/4$.
    Dark green and light stochastic curves show the random walk exhibited by the log-likelihood ratio $Z_1^n$ for different measurement sequences $x_1^n$.  CUSUM algorithm keeps track of $Z_1^n$, starting at $Z_1^0=0$ and updating its value at every time step by $Z_1^n=Z_1^{n-1}+Z(x_n)$ depending on the (random) outcome $x_n$. The algorithm stops and signals a change-point detection as soon as $\max_{j<n}{Z_j^n}\geq h$, i.e. as soon as the log-likelihood at the current time  exhibits a net increase $Z_1^n-Z_1^j=Z_j^n$ larger or equal than $h$ with respect to some point $j$ in the past. This threshold is the only free parameter of the algorithm and  regulates the trade-off between the mean detection delay and the false alarm rate. Two scenarios are show-cased: $\circled{1}$ a large threshold value ($h=22$) reduces the chances of false alarms at the expense of long detection times; $\circled{2}$  a low threshold value ($h'=6$) can detect the change point with a small delay, but has a high risk of producing false alarms. The orange thick line shows the average trend
    (over many trajectories), which is given by
    $\EE_\nu[Z_1^n]=-n D(p\|q)$  before the change point ($k\leq \nu$), while at the change point the slope changes abruptly to $D(q\|p)$. The detection delay $\tau$, i.e. the time required arrive to the threshold $h$ after the change has happened is given by $\tau\sim h/D(q\| p)$. As discussed in the main text the threshold value also fixes the mean false alarm time, i.e. the mean stopping time under the no-change hypothesis: $\EE_{\infty}[T] \sim \ex{h}$.

    Note that at a particular time $n$, the likelihood corresponding to a change point at $j$ is  $P^{(\nu=j)}(x_1^n)=\prod_{i=1}^\nu p(x_i) \prod_{i=\nu+1}^n q(x_i)$ and that corresponding to no change at all 
    is $P^{(\infty)}(x_1^n)=\prod_{i=1}^n p(x_i)$.
    It is then immediate to check that the log-likelihood ratio between these hypotheses, $\log\frac{P^{(\nu=j)}(x_1^n)}{P^{\infty}(x_1^n)}$, is given by $Z_j^n=\sum_{i=j}^k Z(x_i)$. Therefore, the $j$ that attains the maximum $\max_{j<n}{Z_j^n}$ is precisely the point in the past where the change has most likely happened, and the stopping condition is equivalent to fixing a minimum likelihood ratio for such change point to be accepted.
    
  %  \pagebreak
    
   % \twocolumngrid

    {
    \subsection{Expected value of $T_1$}
    In the main text we have used that $\EE_0[Z_1^{T_1}]=\EE_0[\sum_{i=1}^{T_1} Z_i]=\EE_0[Z_1] \EE_0[T_1]$ (see e.g. Eq.~(\ref{waldlim})). In this section we prove that this identity holds whenever $p\neq q$, and, more in general, if the same measurement is applied to a generic probability distribution $q'$ such that $D(q'||p)-D(q'||q)>0$. The proof follows from the standard treatments in~\cite{tartakovsky2014}. First of all, we have the following lemma. 
    \begin{lemma}{Wald identity (Corollary 2.3.1 in~\cite{tartakovsky2014})}\label{Waldtarta}
    Let $Y_1,Y_2...$ be iid random variables with finite mean $\mathbb E[Y_i]=\mu$. Let $S_n=\sum_{i=1}^{n}Y_i$. If $T$ is a stopping time such that $\mathbb{E}[T]<\infty$ we have that
    \begin{equation}
    \mathbb{E}[S_T]=\mu\mathbb{E}[T].
    \end{equation}
    \end{lemma}
    Consider now the stopping time $T_1 = \min\{ k \geq 1 : Z_1^k \geq h \}$, where $Z_{1}^k$, as a function of the outcome $x_i$,  is defined from two distribution $p$ and $q$ as in the text, but the distribution of $x_i$ is given by $p$ before the change and by some other distribution $q'$ after the change. As in the main text, we denote as $\EE_{\infty/\nu}$ the expectation values with respect to some measurement acting sequentially on a sequence of copies of $p$ (the change point never happens, $\nu=\infty$) or the change point sequence for a specific finite $\nu$.
    \begin{lemma}\label{lemmadiff}
    If $D(q'||p)-D(q'||q)>0$, we have 
    \begin{equation}
        \mathbb{E}_0[S_{T_1}]=(D(q'||p)-D(q'||q))\mathbb{E}_0[{T_1}].
    \end{equation}
    \end{lemma}
    \begin{proof}
    The expectation value of $Z_{1}$ is 
    \begin{equation}
    \mathbb E_{q'}[Z_1]=\mathbb E_{q'}[\log \frac{1}{p(x_i)}]- \mathbb E_{q'}[\log \frac{1}{q(x_i)}]=\mathbb E_{q'}[\log \frac{q'(x_i)}{p(x_i)}]-\mathbb E_{q'}[\log \frac{q'(x_i)}{q(x_i)}]=D(q'||p)-D(q'||q),
    \end{equation}
    And the same holds for every $Z_i$.
    Since by assumption $D(q'||p)-D(q'||q)>0$, take $y=(D(q'||p)-D(q'||q))/2>0$. Then we must have that $P(Z_k>y)>\epsilon$ for some $\epsilon>0$, and for any integer $k\geq 1$. Moreover, choose $m$ integer such that $m y> h$, then
    \begin{equation}
    P(Z_1^{m}>m y)>P(Z_1^{m}>m y)\geq P(Z_1>y,...,Z_m>y)\geq \epsilon^m,
    \end{equation}
    implying
    \begin{equation}
    P(T_1\geq k m)=P(Z_1^{j}<h, \forall j< km)\leq P(Z_1^{m}<m y,Z_{m+1}^{2m}<m y, ..., Z_{(m-1)k+1}^{km}<my) \leq (1-\epsilon^m)^k.
    \end{equation}
    Finally
    \begin{equation}
    \EE_{0}[T_1]\leq \sum_{k=1}^{\infty}km P(T_1\geq k m) \leq \sum_{k=1}^{\infty}km (1-\epsilon^m)^k< \infty.
    \end{equation}
    We can thus apply Lemma~(\ref{Waldtarta}) to the sequence $Z_1,Z_2,...$ and the stopping time $T_1$.
    \end{proof}

    \subsection{Measured relative entropy of block measurements}
    
    In the main text, Eq.~\ref{convergencemeasured} we have argued that we can choose $l$ such that $\frac{1}{l}D_{M}(\sigma^{\otimes l}||\rho^{\otimes l})\geq (1-\epsilon)D(\sigma||\rho)$ as a function of $D(\sigma||\rho)$ and $d$, the local dimension, using the bound and the measurement of~\cite{Hayashi_2001}.
        Another bound on the measured relative entropy for i.i.d. states can be obtained from one-shot hypothesis testing. For $\rho,\sigma$ such that $\mathrm{supp}\,\sigma\subseteq \mathrm{supp}\, \rho $, the hypothesis testing relative entropy (for $\epsilon>0$ is defined as
    \begin{equation}
    D_{h}^{\epsilon}(\sigma||\rho)=-\log \left(\min_{0\leq E\leq I}\tr[E\rho]: \tr[E\sigma]\geq 1-\epsilon\right),
    \end{equation}
    and the Rényi relative entropy for $t\in[0,+\infty)\backslash \{1\}$ is defined as
    \begin{equation}
    D_{\alpha}(\sigma||\rho)=\frac{1}{1-\alpha}\log\tr[\sigma^{\alpha}\rho^{1-\alpha}].
    \end{equation}
    From Theorem 3.3 of ~\cite{audenaert2012quantum}, we have
    \begin{equation}
    D_{h}^{\epsilon}(\sigma^{\otimes l}||\rho^{\otimes l})\geq l D(\sigma ||\rho)-\sqrt{l}4\sqrt{2} \log \frac{1}{\epsilon}\log\eta+2\log 2,
    \end{equation}
    with $\eta:=1+2^{D_{3/2}(\sigma||\rho)}+2^{-D_{1/2}(\sigma||\rho)}\leq 2^{D_{3/2}(\sigma||\rho)+2}$.
    
    By choosing the measurement that realizes the minimum with $p:=\tr[E\sigma^{\otimes l}]$, $q:=\tr[E\rho^{\otimes l}]=2^{-D_{h}^{\epsilon/2}(\sigma^{\otimes l}||\rho^{\otimes l})}$, we obtain
    $D_{M}(\sigma^{\otimes l}||\rho^{\otimes l})\geq p \log \frac{p}{q}+(1-p)\log\frac{1-p}{1-q}\geq -(1-\epsilon)\log q-\log 2\geq (1-\epsilon/2)D_{h}^{\epsilon/2}(\sigma^{\otimes l}||\rho^{\otimes l})-\log 2$.
    
    Therefore, we obtain
    \begin{equation}
    \frac{1}{l}D_{M}(\sigma^{\otimes l}||\rho^{\otimes l})\geq (1-\epsilon/2)D(\sigma||\rho)-(1-\epsilon/2)\left(\frac{4\sqrt{2}(D_{3/2}(\sigma||\rho)+2)}{\sqrt{l}}\log\frac{2}{\epsilon}\right)-(1-\epsilon)\frac{\log 2}{l}.
    \end{equation}
    Thus, the minimum $l$ such that $\frac{1}{l}D_{M}(\sigma^{\otimes l}||\rho^{\otimes l})\geq (1-\epsilon)D(\sigma||\rho)$ can be determined as a function of $\epsilon, D(\rho||\sigma),D_{3/2}(\rho||\sigma)$. It is not clear to us if this bound is generally stronger than the one we state in the main text.

    In addition to the measurement of \cite{Hayashi_2001}, which does not depend on $\sigma$, we consider other block measurements which depend on $\sigma$ explicitly. We restrict to qubits for simplicity, but the same can be done for higher dimensional states. 
    $\rho^{\otimes n}$ can be written in block diagonal form according to the eigenspaces of total angular momentum, labeled by a half integer $j$ (see e.g.~\cite{Hayashi_2001}). We get
    \begin{equation}
    \rho^{\otimes l}=\sum_{j=\frac{n\mod 2}{2}}^{l/2} \rho_j^{(l)} \otimes \frac{I}{\nu_j^{(l)}},
    \end{equation}
    where $\nu_j^{(l)}=\binom{l}{l/2-j}\frac{2j+1}{l/2+j+1}$ is the number of equivalent irreducible representations with label $j$, and  $\rho_j^{(l)}$ are positive semidefinite operators. For $\rho=\frac{1+r \sigma_z}{2}$, we have
    \begin{equation}
    \rho_j^{(l)}=\left(\frac{1-r^2}{4}\right)^{l/2-j}\sum_{m=-j}^{j}\left(\frac{1+r}{2}\right)^{j-m}\left(\frac{1-r}{2}\right)^{j+m} \ketbra{}{j,m}{j,m},
    \end{equation}
    where $\ket{}{j,m}$ are eigenvectors of the $z$ component of the angular momentum $J_z$ with eigenvalues $m$. For $\sigma=e^{-i\frac{\theta}{2} \sigma_x}\left(\frac{I+r_2\sigma_z}{2}\right)e^{-i\frac{\theta}{2}\sigma_x}$, we have that $ \sigma^{\otimes n}=\sum_{j=\frac{l\mod 2}{2}}^{l/2} \sigma_j^{(l)} \otimes \frac{I}{\nu_j^{(l)}}$ and $\sigma_j^{(l)}$ is simply obtained rotating the operator at $\theta=0$ with the Wigner matrix $D(e^{-i\frac{\theta}{2} \sigma_x})=e^{-i \theta J_x }$. 
     \begin{equation}
    \sigma_j^{(l)}=e^{-i \theta J_x }\left(\left(\frac{1-r^2}{4}\right)^{l/2-j}\sum_{m=-j}^{j}\left(\frac{1+r}{2}\right)^{j-m}\left(\frac{1-r}{2}\right)^{j+m} \ketbra{}{j,m}{j,m}\right)e^{i \theta J_x }.
    \end{equation}

    As ansatz we consider a measurement that first measures the label $j$ and the multiplicity label, and then measures in an arbitrary rotated basis $\{e^{-i \eta_j J_x }\ketbra{}{j,m}{j,m}e^{i \eta_j J_x }\}$.
    Optimizing the measured relative entropy over $\eta_j$ for each $j$, we get what we call j-angle-optimized measurement in the main text. For $\eta_j=0$ for every $j$, it coincides with the measurement in ~\cite{Hayashi_2001}. These class of measurements can be efficiently realized via Schur sampling~\cite{bacon_efficient_2006,harrow_phd,krovi_efficient_2019}, even in their generalization to arbitrary local dimension $d$. For example, for $d=2$,
    one can first apply weak Schur sampling to measure a value of $j$, uncompute the circuit, apply $e^{i \eta_j \sigma_x/2}$ to every qubit and measure the total $J_z$ (which can also be done locally by measuring in the computational basis).

    Finally, according to ~\cite{berta2017variational}, the  measured relative entropy of two states $\rho$ and $\sigma$ can be computed as
    \begin{equation}
    D_{M}(\sigma||\rho)=\sup_{\omega\geq 0}\tr[\sigma\log \omega]-\tr[\rho\omega]+1,
    \end{equation}
    and the measurement realizing it is a projective measurement on the eigenspaces of the optimal $\omega$. By using the block diagonal form of multi-copy states $\rho^{\otimes l}$, $\sigma^{\otimes l}$, it is clear that $\omega$ can also be chosen to be block diagonal. This expression can be computed numerically using the semidefinite approximation of the matrix logarithm \cite{cvxquad}. In particular, to obtain the points in the plot we have used the MATLAB package %CVX~\cite{cvx, gb08}
    CVXQUAD~\cite{cvxquad}. By the plot in Figure \ref{fig:2}, we see that the $j$-angle-optimized measurement performance is very close to the optimal one as calculated by CVXQUAD. However, the optimization fails for $n$ around 10, likely due to the very large range of order of magnitude of the matrix entries. %likely due to the very large range of order of magnitude of the matrix entries (from $10^{-15}$ to $10^{-4}$ in one inspected case), and gives inaccurate answers for even smaller $n$.

    }
    \subsection{QUSUM for general change point}

    If the change point does not happen at a multiple of $l$, the first inequality of Eq.~(\ref{upperbound0}) still holds, while we can replace the second one with
    \begin{align}
    \bar\tau^{\star}&=\sup_{\nu\geq 0} \sup_{x^\nu: P^{(\infty)}(X^\nu=x^\nu)>0} \EE_{\nu}[lT^*-\nu | lT^* > \nu,X^\nu=x^\nu] \nonumber\\
    &\leq  \sup_{\nu\geq 0} \sup_{x^\nu: P^{(\infty)}(X^\nu=x^\nu)>0} \EE_{0}[l T_{\lfloor\nu/l\rfloor+2}-\nu | T^* > \nu, X^\nu=x^\nu]\nonumber\\
    &=l\EE_{0}[T_2]=l(1+\EE_{0}[T_1]),
    \end{align}
    where the inequality comes from the definition of $T^*$ (Eq.~(\ref{eq:stopbig}) and the equalities from the fact that $T_{\lfloor\nu/l\rfloor+2}$ does not depend on the outcomes before $t=\nu+1$.
    Therefore, we get again Eq.~(\ref{tradeoff}).
    
    We also discuss how we can generalize this achievability result in the case where we have that the state after the change point is unknown and belonging to a family of states $\mathcal{S}$. In this case, we consider the random variable $S_{j}=\sup_{\sigma\in\mathcal{S}}Z_{j}^k(\sigma)$, where $Z_j^k(\sigma)$ is the one from Eq.~(\ref{defzi}), and we generalize the stopping time accordingly, stopping as soon as $S_j\geq h $. In particular we consider the parallel stopping times $T_j^{(\sigma)}$, $j\geq 1$, $\sigma\in \mathcal S$, which are defined as in the main text but depend on the different $\sigma$, and define $T_j=\min_{\sigma\in\mathcal S}T_j^{(\sigma)}$.
    
    Since the POVM ${M^{(k)}_{x_i}}$ for which Eq.~(\ref{convergencemeasured}) holds does not depend on $\sigma$, we get asymptotically (for large blocks and large $h$)
    \begin{align}
    \bar\tau^\star \leq \frac{ h(1+o(1))}{\min_{\sigma\in\mathcal{S}}D(\sigma \| \rho )}.
    \end{align}
    { If the same test is applied to a state $\sigma'$ not necessarily in $\mathcal S$ but such that $\min_{\sigma\in S}D(\sigma'||\rho)-D(\sigma'||\sigma)>0,$ we have asymptotically, for a change point sequence with $\sigma'$ as post-change state
    \begin{align}
    \bar\tau^\star \leq \frac{h(1+o(1))}{\max_{\sigma\in\mathcal{S}}(D(\sigma'||\rho)-D(\sigma'||\sigma))},
    \end{align}
using the convergence of measured relative entropy and Lemma~\ref{lemmadiff}.
    This allows to give guarantee on the delay in the case of unknown post-change state in a continous family $\mathcal{S}'$, if $\mathcal {S}$ is a suitable finite discretization of the family, in the sense that $\inf_{\sigma'\in\mathcal S}\max_{\sigma\in\mathcal{S}}(D(\sigma'||\rho)-D(\sigma'||\sigma))>0$. In the worst case, if the post-change state is in $\mathcal S'$, we have asymptotically
    
    \begin{align}
    \bar\tau^\star \leq \frac{h(1+o(1))}{\inf_{\sigma'\in\mathcal{S}'}\max_{\sigma\in\mathcal{S}}(D(\sigma'||\rho)-D(\sigma'||\sigma))}.
    \end{align}

%    This discretization always exists in finite dimension if $\varepsilon:=\inf_{\sigma'\in\mathcal S}D(\sigma'||\rho)>0$, since 
  %  \begin{itemize}
   % \item An $\epsilon$-cover of $\mathcal{S}'$ in trace distance made of non-pure states exists for any $\epsilon>0$ and any subset
   % \item A reversed Pinsker inequality holds for non-pure states
   % \end{itemize}
    }

    On the other hand, using that $I_{T_1<\infty}=\cup_{\sigma\in \mathcal S}I_{T_1^{(\sigma)}<\infty}$ and the union bound, Eq.~(\ref{boundfalse}) becomes
    \begin{align}\label{boundfalse2}
    P_\infty(T_1<\infty)&\leq \sum_{\sigma\in \mathcal S}P_\infty(T_1^{(\sigma)}<\infty)%\nonumber\\&
    =\sum_{\sigma\in \mathcal S}\mathbb E_0\left[\frac{p}{q(\sigma)} I_{ T_1^{(\sigma)}<\infty}\right]%\nonumber\\
    =\sum_{\sigma\in\mathcal S}\EE_0\left[\ex{-Z^{T_1}_1(\sigma)}I_{ T_1^{(\sigma)}<\infty}\right]\nonumber\\
    &=\sum_{\sigma\in\mathcal S}\mathbb{E}_0\left[\ex{-h-s(\sigma)}I_{ T_1^{(\sigma)}<\infty}\right]\leq |\mathcal S|\ex{-h} =: \alpha,\
    \end{align}
    meaning that we can take $h=\log \bar{T}_{\text{FA}}+\log |\mathcal S|$, obtaining the desired statement as $\bar{T}_{\text{FA}}\rightarrow\infty$. 

    \subsection{Proof of Theorem 3}\label{LAI_proof}
    We restate the theorem and prove it.
    \begin{theorem}
    For a change point model with log-likelihoods $Z_k^{(\nu)}$ and $\epsilon>0$, no strategies can exceed the trade-off given by $\bar\tau^*\geq (1-\epsilon) I^{-1}\log{\bar{T}_{\text{FA}}}(1+o(1))$, for large $\bar{T}_{\text{FA}}$, for any $I$ that satisfies the condition 
    \begin{align}\label{eq:condP2}
    &\lim_{n\rightarrow\infty} \sup_{\nu\geq 0} \sup_{x^\nu: P^{(\infty)}(X^\nu=x^\nu)>0}
 P^*_\nu(x^\nu) =0 %\quad \forall\delta>0
    \mbox{ where }\\
    P^*_\nu(x^\nu)&:=P_{\nu}\left\{ \max_{ t\leq n }  \lambda^{(\nu)}_{\nu+t} \geq I(1+\epsilon) n \middle\vert X^\nu=x^\nu \right\} \nonumber
    \end{align} 
    \end{theorem}

    \begin{proof}
    First notice that for any $\nu\geq0$
    \begin{equation}\label{eq:taubound}
    \bar{\tau}^*\geq \bar{\tau} \geq \EE_{\nu}[T-\nu | T > \nu],
    \end{equation}
    
    and by Markov inequality
    \begin{equation}\label{eq:markov}
    \EE_{\nu}[T-\nu | T > \nu]\geq m P_{\nu}(T-\nu\geq m |T>\nu)  
    \end{equation}
    
    We start from observing that, for the event $E_{\nu}:=\nu< T<\nu+m$ we have
    \begin{equation}\label{eq:event}
    P_{\infty}(E_{\nu})=\mathbb{E}_{\nu}[I_{E_{\nu}}\ex{-\lambda^{(\nu)}_{T}}],
    \end{equation}
    where 
    \begin{align}
    \lambda^{(\nu)}_{T}&=\log \frac {\tr[M_{x_1,...,x_T}(\rho^{\otimes \nu-1}\otimes \sigma^{\otimes T-\nu})]}{\tr[M_{x_1,...,x_T}\rho^{\otimes T}]}, &\mathrm{if}\,\,T> \nu, \nonumber\\
    \lambda^{(\nu)}_{T}&=0,  &\mathrm{if}\,\,T\leq\nu,
    \end{align}
    and the change of measure is justified since $D_{\max}(\sigma||\rho)<\infty$.
    Then we have
    \begin{align}
    \mathbb{E}_{\nu}[I_{E_{\nu}}\ex{-\lambda^{(\nu)}_{T}}]&\geq \mathbb{E}_{\nu}[I_{\{E_{\nu}, \lambda^{(\nu)}_{T}< c\}}\ex{-\lambda^{(\nu)}_{T}}]%\nonumber\\ &
    \geq \ex{-c}P_{\nu}(E_{\nu}, \lambda^{(\nu)}_{T}< c)
   % \\&
    \geq \ex{-c}P_{\nu}(E_{\nu}, \max_{\nu< n< \nu+m}\lambda^{(\nu)}_{n}< c)\nonumber\\
    &\geq \ex{-c}P_{\nu}(E_{\nu})-\ex{-c}P_{\nu}\left(\max_{\nu< n< \nu+m}\lambda^{(\nu)}_{n}\geq c\right),
    \end{align}
    where the last inequality comes from the fact that for any two events $A$ and $B$, $P(A\cap B)\geq P(A)-P(B^c)$.
    From this chain of inequalities and Eq.~(\ref{eq:event}) we get
    \begin{align}
    P_{\nu}(\nu< T<\nu+m)&\leq \ex{c} P_{\infty}(T<\nu+m)%\nonumber\\ &
    +P_{\nu}\left(\max_{\nu\leq n< \nu+m}\lambda^{(\nu)}_{n}\geq c\right)
    \label{eq:ineq1}
    \end{align}
    This inequality is valid also if we substitute probabilities with conditional probabilities with respect to the event $T>\nu$. In order to see this, we observe  that  $P_{\infty}(T>\nu)=1-P_{\infty}(T\leq\nu)=1-P_{\nu}(T\leq\nu)=P_{\nu}(T>\nu)$, since the probability of stopping at $\nu$ depends only on measurement outcomes up to $\nu$, whose distribution is identical for the cases where the change point is at $\nu$ or at infinity. Moreover, $P_{\infty}(T> \nu)>0$ if we require $\bar{T}_{\text{FA}}$ sufficiently large. We can then divide both members of the inequality \eqref{eq:ineq1} by $P_{\infty}(T> \nu)$ and get by Bayes' rule
    \begin{align}
    P_{\nu}(T<\nu+m|T>\nu)\leq \ex{c} P_{\infty}(T<\nu+m|T>\nu) %\nonumber\\
    +P_{\nu}\left(\max_{\nu\leq n< \nu+m}\lambda^{(\nu)}_{n}\geq c|T>\nu\right)
    \label{eq:ineq2}
    \end{align}
    
    We next show that the RHS, and hence the conditional probability go to zero as we increase $m$. We first concentrate on the first term of the right hand side. For this purpose, we write
    \begin{align}
    \bar{T}_{\text{FA}}&=\mathbb{E}_{\infty}[T]=\sum_{j=0}^{\infty}P_{\infty}(T> j)%\nonumber\\&
    =\sum_{i=0}^{m-1}\sum_{k=0}^{\infty}P_{\infty}(T> i+km)\nonumber\\
    &=\sum_{i=0}^{m-1}\sum_{k=0}^{\infty}P_{\infty}(T> i)P_{\infty}(T> i+km|T> i)I_{P_{\infty}(T> i)>0}\end{align}
    
    Defining 
    \begin{equation}
    \Delta:=\sup_{i\geq 0: P_{\infty}(T>i)>0}P_{\infty}(T> i+m|T> i),
    \end{equation}
    we have that if $P_{\infty}(T>i)>0$ either $P_{\infty}(T> i+km|T> i)=0$ or
    \begin{align}\label{trick}
   & P_{\infty}(T> i+km|T> i)=P_{\infty}(T> i+(k-1)m+m|T> i)
    \nonumber\\
    &=P_{\infty}(T> i+(k-1)m+m|T> i+(k-1)m)P(T> i+(k-1)m|T> i)\nonumber\\
   & = \prod_{j=0}^{k-1}P_{\infty}(T> i+j m+m|T> i+j m)%\nonumber\\
    \leq \Delta^k.
    \end{align}
    
    Therefore
    \begin{align}
    \bar{T}_{\text{FA}}&\leq \sum_{i=0}^{m-1}\sum_{k=0}^{\infty}P_{\infty}(T>i)\Delta^k%\nonumber\\
    =\frac{1}{1-\Delta}\sum_{i=0}^{m-1}P_{\infty}(T>i)\leq \frac{m}{1-\Delta}
    \nonumber\end{align}
    
    From which we obtain $\Delta\geq1- m/\bar{T}_{\text{FA}}$. If the superior in the definition of $\Delta$ is actually a maximum, then we can choose $\nu$ such that
    \begin{align}
    &P_{\infty}(T\leq\nu+m|T> \nu)=1-P_{\infty}(T>\nu+m|T> \nu)%\nonumber\\
    = 1-\Delta \nonumber \leq m/\bar{T}_{\text{FA}}.
    \end{align}
    
    Otherwise, the last inequality in Eq.~(\ref{trick}) is strict and we get $\Delta>1-m/\bar{T}_{\text{FA}}$. Since $\Delta>1-m/\bar{T}_{\text{FA}}$ implies that there exists $\nu$ such that $P_{\infty}(T>\nu+m|T> \nu)>1-m/\bar{T}_{\text{FA}}$, and for the same $\nu$ it holds that $P_{\infty}(T>\nu+m|T> \nu)\geq 1-m/\bar{T}_{\text{FA}}$.
    
    We now take $m=(1-\epsilon)I^{-1} \log \bar{T}_{\text{FA}} $, and $c=(1+\epsilon)m I=(1-\epsilon^2)\log\bar{T}_{\text{FA}}$ and upper-bound the first term in RHS of inequality \eqref{eq:ineq2} by
    \begin{align}
    \ex{c} P_{\infty}(T<\nu+m|T>\nu)%\nonumber\\&
    \leq (\bar{T}_{\text{FA}})^{1-\epsilon^2}\frac{(1-\epsilon)I^{-1}\log \bar{T}_{\text{FA}}}{\bar{T}_{\text{FA}}}%\nonumber\\
    =\frac{(1-\epsilon)I^{-1}\log \bar{T}_{\text{FA}}}{\bar{T}_{\text{FA}}^{\epsilon^2}}
    \end{align}
    
    Putting all together,  we conclude that the LHS of \eqref{eq:ineq2} is upper-bounded by a quantity that goes to zero:
    \begin{align}
    P_{\nu}(T<\nu+m|T>\nu)\leq & \frac{(1-\epsilon)I^{-1}\log \bar{T}_{\text{FA}}}{\bar{T}_{\text{FA}}^{\epsilon^2}} \nonumber
    \\ &+P_{\nu}\left(\max_{\nu\leq n< \nu+m}\lambda^{(\nu)}_{n}\geq (1+\epsilon)m I|T>\nu\right)%\nonumber\\&
    \rightarrow 0, \qquad (m\rightarrow \infty)
    \end{align}
    
    The first term goes to zero since $(\log \bar{T}_{\text{FA}})/{\bar{T}_{\text{FA}}^{\epsilon^2}} $ goes to zero as $m$ goes to infinity by definition of $m$, and the second term goes to zero by hypothesis of the theorem.
    
    Recalling Eq.~(\ref{eq:markov}), this means that there exists $\nu$ such that
    \begin{equation}
    \EE_{\nu}[T-\nu | T > \nu]\geq (1-\epsilon)\frac{\log{\bar{T}_{\text{FA}}}}{I}(1+o(1)),
    \end{equation}
    and therefore, by Eq.~(\ref{eq:taubound})
    \begin{equation}
    \bar{\tau}^*\geq \bar{\tau}\geq (1-\epsilon)\frac{\log{\bar{T}_{\text{FA}}}}{I}(1+o(1)),
    \end{equation}
    \end{proof}
    \subsection{Change point with channels}
    
    In this section we will extend the previous results to quantum channels. From now on we consider a process that is modeled by a sequences of quantum channel $\mathcal N_{A_i\rightarrow B_i}^{(i)}$ which can be either $\mathcal{M}_{A_i\rightarrow B_i}$ (if $i>\nu$) or $\mathcal{N}_{A_i\rightarrow B_i}$ (if $i\leq\nu$). Let us first consider a modification of QUSUM to show an achievable rate. For every instance of the process we input a state $\rho$ into the channel, leading to an output state $\rho_i=\NN^{(i)}(\rho)$. We can now simply run the previously used state version of QUSUM on these output states and generalizing Equation~\eqref{Eq:meas-RE-achievable} we get
    \begin{align}
    \bar\tau^\star \leq \frac{\log\bar{T}_{\text{FA}}}{D_M(\mathcal{M} \| \mathcal{N} )}+O(1),
    \end{align}
    where
    \begin{align}
        D_M(\mathcal{M} \| \mathcal{N} ) = \sup_\rho D_M(\mathcal{M}(\rho) \| \mathcal{N}(\rho)), 
    \end{align}
    by optimizing over input states $\rho$. The same way collective measurements lead to an advantage in the state case, also entangled input states can improve the performance. Using joint states on $l$ systems combined with collective measurements on the same, we get
    \begin{align}\label{channel-tradeoff}
    \bar\tau^\star \leq \frac{\log\bar{T}_{\text{FA}}/l}{\frac1l D_M(\mathcal{M}^{\otimes l} \| \mathcal{N}^{\otimes l} )}+O(1).
    \end{align}
    Finally we observe that 
    \begin{align}\label{Eq:channel-DM-conv}
        \lim_{l\rightarrow\infty}\frac1l D_M(\mathcal{M}^{\otimes l} \| \mathcal{N}^{\otimes l} ) = D^\infty(\mathcal{M}\|\mathcal{N}), 
    \end{align}
    where
    \begin{align}
        D^\infty(\mathcal{M}\|\mathcal{N}) = \lim_{l\rightarrow\infty} \sup_\rho \frac1l D(\mathcal{M}^{\otimes l}(\rho) \| \mathcal{N}^{\otimes l}(\rho) ). 
    \end{align}
    
    This can be seen as follows. The $\leq$ direction is a simple consequence of data processing. For the $\geq$ direction fix $l=mk$ and consider
    \begin{align}
        &\frac1{mk} D_M(\mathcal{M}^{\otimes {mk}} \| \mathcal{N}^{\otimes {mk}} ) 
        \geq \sup_\rho \frac1{mk} D_M((\mathcal{M}^{\otimes m}(\rho))^{\otimes k} \| (\mathcal{N}^{\otimes m}(\rho))^{\otimes k} ). 
    \end{align}
    
     Taking first the limit $k\rightarrow\infty$ gives us the relative entropy, then taking $m\rightarrow\infty$ gives the desired regularization. We remark here that the regularized channel relative entropy is known to be the optimal achievable rate in the Stein's Lemma for quantum channels~\cite{wilde2020amortized,wang2019resource,fang2020chain} and that the regularization is generally necessary~\cite{fang2020chain}. 
     Overall, this leads to the asymptotic achievable tradeoff 
     \begin{align}\label{channel-achievable}
    \bar\tau^\star \leq \frac{\log\bar{T}_{\text{FA}}(1+o(1))}{ D^\infty(\mathcal{M} \| \mathcal{N} )}.
    \end{align}
    
    For the converse bound, we again mostly follow the state case, with the difference that the set of allowed strategies is larger. In the case of channels, any sequential strategy can be described as follows: first, we are allowed to prepare any initial state $\rho^{(0)}_{LA_1}$ (possibly infinite dimensional); assume that at step $i$ we have a access to a state $\rho^{x^{i-1}}_{LA_{i}}$, first apply the channel $\mathcal N_{A_i\rightarrow B_i}^{(i)}$ to $\rho^{x^{i-1}}_{LA_i}$, and then we apply the instrument $\mathcal E^{x^{i-1}}$, corresponding to CP maps  $\mathcal E^{x_{i}}_{LB_i\rightarrow LA_{i+1}}$ (where we drop the dependence on $x^{i-1}$ for readability); the state conditioned on the new outcome $x_i$ is then denoted as $\rho^{x^i}_{L A_{i+1}}$. 
    
    The probability of getting outcome $x^i$ is then recursively defined as 
    \begin{align}
    p(x_i|x^{i-1})&=\tr_{LA_{i+1}}[\mathcal E^{x_{i}}_{LB_i\rightarrow L A_{i+1}}\circ%\nonumber\\&\circ 
    \mathcal N_{A_i\rightarrow B_i}[\rho^{x^{i-1}}_{LA_i}]],
    \end{align}
    for $i\leq\nu$ or if there is no change point, and
    \begin{align}
    q^{(\nu)}(x_i|x^{i-1})&=\tr_{LA_{i+1}}[\mathcal E^{x_{i}}_{LB_i\rightarrow L A_{i+1}}%\nonumber\\
    \circ 
    \mathcal M_{A_i\rightarrow B_i}^{(i)}[\rho^{x^{i-1}}_{LA_i}]],
    \end{align}
    if the change point is at $\nu$ and $i>\nu$.
    
\medskip
    
    Using the same notation as in the state case, we get
    \begin{align}
        &P^{(0)}\left(\max_{1\leq i\leq n} \lambda_{x^i}\geq n I (1+\epsilon)\right) %\\
        = \sum_{\substack{x^n: \\ \max_{1\leq i\leq n} \lambda_{x^i}\geq n I(1+\epsilon)}} q^{(0)}(x^n) \\
        &= \sum_{1\leq i \leq n} \sum_{x^i\in S_i} q^{(0)}(x^i)= \sum_{1\leq i \leq n} \sum_{x^i\in S_i} \tilde q^{(0)}(x^i) 
    \end{align}
    where $\tilde q^{(0)}(x^i)$ are the outcomes probabilities corresponding to a modified strategy, where upon obtaining an output $x^i\in S_i$, a fixed state is send through the next channels and the output is discarded.  Clearly, $\tilde q^{(0)}(x^i)$ can  still be defined recursively as  $q^{(0)}(x^i)$.

    At fixed $n$, the test obtained by accepting if the strategy finds $x^i\in S_i$ at some step $i$ is a valid binary test, which has probability of success when performed on $n$ accesses to $\mathcal{M}$ equal to $P^{(0)}\left(\max_{1\leq i\leq n} \lambda_{x^i}\geq n  I(1+\epsilon)\right)$. 
    
    In the same way we define $\tilde p(x^i)$, obtained applying the same strategy to accesses to $\mathcal N$. 
    By construction, we have $\tilde q^{(0)}(x^i)\geq \ex{n I(1+\epsilon)}\tilde p(x^i)$ for $x^{i}\in S_i$.  A strong converse bound resulting from~\cite{fawzi2021defining} applies to this strategy and the pair $(\mathcal M,\mathcal N)$, and the same argument of the state case gives that if $I\geq \tilde D_1^\infty(\mathcal{M}\|\mathcal{N})$, $\lim_{n\rightarrow \infty }P^{(0)}\left(\max_{1\leq i\leq n} \lambda_{x^i}\geq n  I(1+\epsilon)\right)=0$.
    
    However, note that we do not know whether $\tilde D^\infty_1(\mathcal{M}\|\mathcal{N}) = D^\infty(\mathcal{M}\|\mathcal{N})$, but we conjecture that they are equal.
    Combining the above with the QUSUM optimality condition we have that, asymptotically, for every $\epsilon>0$,
    \begin{align}
        \bar\tau^\star \geq (1-\epsilon)\frac{\log\bar{T}_{\text{FA}}(1+o(1))}{\tilde D^\infty_1(\mathcal{M}\|\mathcal{N})}. 
    \end{align}
    Our achievability and optimality results give close bounds on the asymptotic performance and they do indeed match if the aforementioned conjecture holds. 
   
    \end{document}